\documentclass[11pt]{article}

\usepackage[utf8x]{inputenc}
\usepackage[T1]{fontenc}

\usepackage{amsmath, amsfonts, amssymb, amsthm}
\usepackage{mathrsfs}

\usepackage{bbm}
\usepackage{latexsym}
\usepackage{textcomp}

\usepackage{tensor}

\usepackage{graphicx}
\usepackage{subcaption}

\usepackage{cite}

\usepackage[a4paper, textwidth=17cm,textheight=22cm]{geometry}

\usepackage{todonotes}

\usepackage[affil-it]{authblk}

\allowdisplaybreaks

\DeclareMathOperator{\pa}{\partial}

\newcommand{\oo}{\circ}

\newcommand{\e}{\varepsilon}
\newcommand{\m}{\mu}
\newcommand{\n}{\nu}
\newcommand{\al}{\alpha}
\newcommand{\be}{\beta}
\newcommand{\de}{\delta}
\newcommand{\De}{\Delta}

\newcommand{\ph}{\phi}

\newcommand{\Ph}{\Phi}

\newcommand{\g}{\gamma}
\newcommand{\G}{\Gamma}

\newcommand{\La}{\Lambda}

\newcommand{\x}{\xi}
\newcommand{\Si}{\Sigma}
\newcommand{\si}{\sigma}
\newcommand{\om}{\omega} 
\newcommand{\Om}{\Omega}
\newcommand{\te}{\theta}

\newcommand{\z}{\zeta}
\newcommand{\ta}{\tau}

\newcommand{\CD}{\mathcal{D}}

\newcommand{\CT}{\mathcal{T}}

\newcommand{\CP}{\mathcal{P}}

\newcommand{\uvd}{\underline{\mathrm{deg}}}

\newcommand{\WF}{\mathrm{WF}}
\newcommand{\bydef}{\stackrel{.}{=}}
\newcommand{\oli}{\overline}
\newcommand{\ul}{\underline}

\newcommand{\na}{\nabla}

\newcommand{\IR}{\mathbb{R}}
\newcommand{\IC}{\mathbb{C}}

\newcommand{\II}{\mathbbm{1}}

\newcommand{\IP}{\mathbb{P}}

\newcommand{\SF}{\mathscr{F}}

\newcommand{\SCA}{\mathscr{A}}

\newcommand{\SCB}{\mathscr{B}}
\newcommand{\SCL}{\mathscr{C}}
\newcommand{\SD}{\mathscr{D}}
\newcommand{\ST}{\mathscr{T}}

\newcommand{\SV}{\mathscr{V}}
\newcommand{\SE}{\mathscr{E}}
\newcommand{\SZ}{\mathscr{Z}}

\DeclareFontFamily{U}{mathx}{\hyphenchar\font45}
\DeclareFontShape{U}{mathx}{m}{n}{
      <5> <6> <7> <8> <9> <10>
      <10.95> <12> <14.4> <17.28> <20.74> <24.88>
      mathx10
      }{}
\DeclareSymbolFont{mathx}{U}{mathx}{m}{n}
\DeclareFontSubstitution{U}{mathx}{m}{n}
\DeclareMathAccent{\widecheck}{0}{mathx}{"71}

\DeclareFontFamily{U}{mathb}{\hyphenchar\font45}
\DeclareFontShape{U}{mathb}{m}{n}{
      <5> <6> <7> <8> <9> <10>
      <10.95> <12> <14.4> <17.28> <20.74> <24.88>
      mathb10
      }{}
\DeclareSymbolFont{mathb}{U}{mathb}{m}{n}
\DeclareFontSubstitution{U}{mathb}{m}{n}
\DeclareMathSymbol{\lineco}{2}{mathb}{'156}

\DeclareFontFamily{U}{matha}{\hyphenchar\font45}
\DeclareFontShape{U}{matha}{m}{n}{
<5>matha5<6>matha6<7>matha7<8>matha8<9>matha9
<10><10.95>matha10
<12><14.4><17.28><20.74><24.88>matha12
}{}
\DeclareSymbolFont{matha}{U}{matha}{m}{n}
\DeclareFontSubstitution{U}{matha}{m}{n}
\DeclareMathSymbol{\ovoid}{\mathbin}{matha}{"6C}

\theoremstyle{plain}
\newtheorem{thm}{Theorem}
\newtheorem{lem}{Lemma}
\newtheorem{prop}{Proposition}
\newtheorem{cor}{Corollary}

\theoremstyle{definition}
\newtheorem{defi}{Definition}

\theoremstyle{remark}
\newtheorem{rmk}{Remark}

\title{Normal Products and Zimmermann Identities in Configuration Space BPHZ Renormalization}
\author[1,2]{Steffen Pottel \thanks{Electronic address: steffen. pottel 'at' itp. uni-leipzig. de}}
\affil[1]{\small Institut f\"ur Theoretische Physik, Universit\"at Leipzig, Postfach 100920, D-04009 Leipzig, Germany }
\affil[2]{\small Max-Planck-Institute for Mathematics in the Sciences, Inselstra\ss e 22, D-04103 Leipzig, Germany}
\date{\today}

\begin{document}

\maketitle

\begin{abstract}\begin{center}\begin{minipage}{0.8\textwidth}
	\noindent
	The notion of normal products, a generalization of Wick products, is derived with respect to BPHZ renormalization formulated entirely in configuration space. 
	Inserted into time-ordered products, normal products admit the limit of coinciding field operators, which constitute the product.	
	The derivation requires the introduction of Zimmermann identities, which relate field monomials or renormalization parts with differing subtraction degree. 
	Furthermore, we calculate the action of wave operators on elementary fields inserted into time-ordered products, exploiting the properties of normal products.
	\end{minipage}\end{center}
\end{abstract}

\maketitle

\setlength{\parindent}{0em}
\setlength{\parskip}{1.0ex plus 0.5ex minus 0.2ex}

\section{Introduction}

In perturbative quantum field theory, most physical quantities are ill-defined and the approximation by a formal power series about the free theory turns out to be too rough already at finite order.
A constructive way to extract physically reasonable information can be found in renormalization theory.
There are many renormalization schemes available and the fact that all of them are equivalent does not suggest that one or the other scheme makes the renormalization of a perturbatively treated quantum field theory significantly easier.
Usually, popular schemes are preferred for certain aspects or are more convenient in treating certain problems.
For instance, the Epstein-Glaser scheme \cite{Epstein:1973gw} or analytic renormalization \cite{Speer:1972wz} maintain causality in the construction.
Dimensional renormalization \cite{Bollini:1972ui,'tHooft:1972fi} is widely used for explicit computations.
And the BPHZ scheme \cite{Bogoliubov:1957gp,Hepp:1966eg,Zimmermann:1968mu,Zimmermann:1969jj} for massive fields as well as the BPHZL scheme \cite{Lowenstein:1975rg,Lowenstein:1975ps,Lowenstein:1975ku} if additionally massless fields are included, are often applied in studies regarding the structural properties of a quantum field theory.  
Its probably most prominent feature is the forest formula, which resolves the combinatorial structure of singularities found in the loop structure of weighted Feynman graphs.
The formula can be found in various approaches to renormalization \cite{Breitenlohner:1975hg,Breitenlohner:1976te,Keller:2010xq,Hollands:2010pr} and its deeper mathematical structure has been established in the realm of Hopf algeras \cite{Connes:1999yr,Connes:2000fe,EbrahimiFard:2010yy,Hairer:2017nmk}, where regularization have been used, which differ from Zimmermann's approach. 
It turns out that the regularization method determines the renormalization parts, i.e. parts of the weighted Feynman graph which require renormalization, and thus influences the proof of the forest formula.
In Bogoliubov's approach, a variation of Hadamard regularization for singular integrals is used.
The variation is necessary due to the formulation of quantum field theories with coupling constants \cite{Itzykson:1980rh}, i.e. not at each vertex of the Feynman graph there is a test function at disposal, and it consists of subtracting Taylor polynomials of the graph weights, which must belong to the complement of the renormalization part.
This choice of regularization entails a modification of the graph structure, once the subtractions are performed, which can be of use for various applications.
With the reduction formalism at hand, it is possible to define composite operators of interacting quantum fields out of perturbation theory \cite{Zimmermann:1972te}. 
Moreover, the scheme admits the assignment of scaling dimensions to such composite operators, which are greater or equal to the naive scaling dimension. 
This may lead to a change in the $R$-operation for affected weighted Feynman graphs and is sometimes called oversubtraction. 
Then it is quite natural to ask whether different assignments can be related to each other. 
The positive answer to this question is given by the Zimmermann identity \cite{Zimmermann:1972te}, which establishes that two choices of degrees for the same composite operator differ only by a finite sum of other composite operators with well-defined degrees. 
The identity \cite{Clark:1976ym} as well as the scheme \cite{Gomes:1974cr} can be generalized in order to derive the equation of motion for specific quantum fields and study symmetry breaking in the sense that the breaking of Ward-Takahashi identities is given by an insertion of a composite operator into the correlation functions order by order in Planck's constant \cite{Lowenstein:1971jk}. 
This property was used for BRST quantization \cite{Becchi:1975nq,Tyutin:1975qk} and in regard to parametric differential equations \cite{Zimmermann:1979fd}, where reviews on the technique can be found in \cite{Piguet:1980nr,Piguet:1986ug} and some illustrating examples are provided by \cite{Kraus:1991cq,Kraus:1992ru,Kraus:1997bi,Pottel:2010hp}. 
Having well-defined composite operators offers another application in view of coincidence limits of quantum fields, where products of quantum fields can be expressed by local (composite) fields multiplied by structure functions, which capture the singular behavior of the initial fields approaching each other in spacetime \cite{Wilson:1972ee}. 
With a generalization of Wick ordering of quantum fields, named normal products, it is then possible to prove the operator product expansion in perturbation theory \cite{Zimmermann:1972tv}.

In recent years, many renormalization techniques have been revisited in view of the progress in the formulation of quantum field theory on curved spacetimes \cite{Hollands:2014eia,Fredenhagen:2015iia}.
In this setting, the Epstein-Glaser method has been extended \cite{Brunetti:1999jn,Hollands:2001nf,Hollands:2001fb} and some of its aspects were studied and improved \cite{Bahns:2012pw,GraciaBondia:2002js,GraciaBondia:2002gq,Falk:2009ug}.
Other schemes were translated to configuration space \cite{Duetsch:2013xca} and formulated for non-trivial geometries \cite{Gere:2015qsa,Dang:2017gma}.
Also the BPHZ scheme was modified \cite{Zavyalov:1990kv} and related to configuration space by either Fourier transformation of Epstein-Glaser renormalization \cite{GraciaBondia:2002js,GraciaBondia:2002gq,Prange:1997iy} or by Fourier transformation of the graph weights as functions for quantum electrodynamics \cite{Steinmann:2000nr}.
The latter approach is picked up and generalized to general Feynman graphs in \cite{Pottel:2017aa}.
In \cite{Pottel:2017bb}, the BPHZ method was formulated in configuration space, extending its applicability to quantum field theories defined on analytic spacetimes. 
In this regard, it is natural to derive normal products, a generalization of Wick products, in the configuration space prescription. 
Inserted into correlation functions, those quantities remain finite in the limit of coinciding arguments.
If this property would hold already for Wick products, there would be no need for renormalization.
Since the formulation of the scheme is based on Wick products, normal products have to be viewed as a tool derived from the BPHZ prescription. 
We observe that renormalization parts may change their singular behavior in the limit of coinciding arguments so that we are confronted with the problem of relating renormalization parts of different degrees, where, indeed, increasing the degree of a renormalization part does not break the effect of the $R$-operation. 
Following this observation, we show that Wick monomials with differing degrees assigned to them are related by the configuration space version of the Zimmermann identity and, using the Zimmermann identity, we prove that the limit of coinciding arguments exists for normal products, which have to be defined recursively due to the structure of renormalization parts. 
Finally, we discuss Wick monomials containing the wave operator and derive an equation of motion for the quantum field in perturbation theory. 

The paper is organized as follows. After reviewing the notions of BPHZ renormalization in configuration space, we derive Zimmermann identy in the third and normal products in the fourth section. The wave equation is studied in Section 5. Finally, we draw some conclusions and relate the results to on-going research.

\section{Configuration Space BPHZ Renormalization}

Consider a scalar field $\ph$ fulfilling the linear equation of motion
\begin{align}
	P\ph(x)=0 
\end{align}
on a four-dimensional globally hyperbolic analytic spacetime $(M,g)$ with Lorentzian metric $g$ and with $P$ being a normally, hyperbolic differential operator of second order. 
Promoting the field to a distribution $\ph(f)$ with $f\in\CD(M)$, it is used to generate the free, unital $*$-algebra $\SCA(M,g)$, which satisfies the conditions
\begin{align}
	&\ph(f)^* - \ph(\oli f) = 0, \\
	&\ph(Pf) = 0, \\
	&\ph(a f_1 + b f_2) - a \ph(f_1) - b \ph(f_2) = 0 \mbox{ with } a,b\in \IC, \\
	&[\ph(f_1),\ph(f_2)] - iF(f_1,f_2) \II = 0,
\end{align}
where $F$ is the commutator function defined as the difference of advanced and retarded fundamental solution. 
A state 
\begin{align}
	\om: \SCA(M,g) \rightarrow \IC 
\end{align}
on $\SCA(M,g)$ is said to be Hadamard if the singularity structure of the two-point functions is fully contained in the Hadamard parametrix 
\begin{align}\label{eq:HadamardParametrix}
	H(x,y) = \frac{1}{4\pi^2} \left[ \frac{U(x,y)}{\si(x,y)} + V(x,y) \log\left(\frac{\si(x,y)}{\La}\right) \right],
\end{align}
where $V$ is a formal power series with finite radius of convergence and $\si$ denotes the squared geodesic distance. 
This notion admits the state-independent definition of Wick ordering in a geodesically convex region $\Om\subset M$, which is recursively given by \cite{Hollands:2001nf}
\begin{align}
	:\ph(f):_H                           & \bydef \ph(f) \nonumber\\
	:\ph(f_1)...\ph(f_n):_H \ph(f_{n+1}) & = :\ph(f_1)...\ph(f_{n+1}):_H  \nonumber \\
	  & \qquad + \sum_{j=1}^n :\ph(f_1)...\widecheck{\ph(f_j)}...\ph(f_n):_H H(f_j,f_{n+1}),\label{eq:WickOrdering} 
\end{align}
where $\widecheck\bullet$ denotes the extraction of that field from the Wick polynomial. 
This prescription of ordering admits the limit $f_j\rightarrow f$ and, denoting by $\CP(x)$ a polynomial in the metric $g$, the Riemann tensor $R_{abcd}$, its symmetrized covariant derivatives as well as the mass $m$ and the coupling $\x$ to the curvature, a generalized Wick monomial can be written as 
\begin{align}\label{eq:CurvPoly}
	:\Ph(x):_{H} \bydef :\CP(x) \prod \na_{(f_1}...\na_{f_l)}\ph(x):_{H}
\end{align}
and used to define the algebra of field observables
\begin{align}
  	\SCB(M,g) \bydef \left\{ :\Ph:_{H} (f) | f\in\CD(M) \right\}. 
\end{align}
Furthermore, we introduce the notion of time-ordered products
\begin{align}\label{eq:NaiveTimeOrdering}
	\CT(:\Ph(x):_H) & \bydef :\Ph(x):_H ,\\
	\CT(:\Ph(x):_H \cdot :\Ph(y):_H) & \bydef 
	\begin{cases} 
     	:\Ph_1:_{H} (x) \cdot :\Ph_2:_{H} (y) \quad \mbox{  for  } \quad x\notin J^-(y)\\
     	:\Ph_2:_{H} (y) \cdot :\Ph_1:_{H} (x) \quad \mbox{  for  } \quad y\notin J^-(x),
   	\end{cases}
\end{align}
where higher orders are recursively defined via \eqref{eq:WickOrdering}. 
Their domain is restricted to the complement of the union over all sets of coinciding arguments, where the latter is called the large diagonal. 
The extension of the time-ordered product to the diagonal in a physically reasonable way is the objective of renormalization theory and we want to employ the method developed in \cite{Pottel:2017bb}. 
For this purpose, recall that $(M,g)$ is globally hyperbolic such that $M$ is isometric to $\IR \times \Si$, i.e. a foliation of Cauchy surfaces $t\times\Si$ parametrized by $t\in\IR$.
With this, the metric is given by $g=\be dt^2 - g_t$, where $\be$ is smooth and positive, and $g_t$ is a Riemannian metric on $\Si$ depending smoothly on $t$.
Then the first step in the construction of the BPHZ scheme is the analytic continuation of the metric
\begin{align}\label{eq:MetricAnaCont}
	g^\e \bydef (1-i\e) \be dt^2 - g_t,\quad \e>0 \, . 
\end{align}
Defining a Riemannian metric $g^R = \be dt^2 + g_t$, it follows that, in a geodesically convex region $\Om\subset \IR\times\Si$, for any $x\in\Om$ and $\x\in T_x\Om$ the inequality
\begin{align}\label{eq:MetricBounds}
	\check C(\e) g^R(\x,\x) \leq |g^\e(\x,\x)| \leq \hat C(\e) g^R(\x,\x)
\end{align}
holds, where $\check C(\e)$ and $\hat C(\e)$ are positive constants for fixed $\e$.
The analytic continuation of the metric entails a consistent change in the differential operator $P\rightarrow P_\e$ as well as the Hadamard parametrix $H\rightarrow H_\e$, i.e. $\si\rightarrow\si_\e$, $U\rightarrow U_\e$ and $V\rightarrow V_\e$, so that 
\begin{align}
	P_\e H_\e = 0
\end{align}
still holds.
Returning to the definition of time-ordered products, we introduce the global time function $T$, which is the projection onto the first component of $\IR\times\Si$, and define the Feynman propagator 
\begin{align}
	H_F(x,y) &= \lim_{\e\rightarrow0^+}(\te(T(x)-T(y))H_\e(x,y) + \te(T(y)-T(x))H_\e(y,x)) \\
	H_{F,\e}(x,y) &= (\te(T(x)-T(y))H_\e(x,y) + \te(T(y)-T(x))H_\e(y,x)),
\end{align}
where $\te$ denotes the Heaviside step function and $H_F(x,y) \in \CD'(M\times M\setminus \mbox{''diagonal''})$.
We observe that
\begin{align}
	P_\e H_{F,\e} = \de
\end{align}
with the character of the differential operator 
\begin{align}
	\mathrm{char}(P_\e) = \emptyset
\end{align}
following from \eqref{eq:MetricBounds}.
Then we have \cite{hormander1990analysis}
\begin{align}\label{eq:WFAreduction}
	\WF_A(H_{F,\e}) &\subseteq \mathrm{char}(P_\e) \cup \WF_A(\de) \\
	&= \WF_A(\de),
\end{align}
which implies that pointwise products of Feynman propagators $H_{F,\e}\in\CD'((\Om\times\Om)\setminus \mbox{``diagonal''})$ exist in a geodesicaly convex region $\Om\subset M$. 
Two remarks are in order.
First, \eqref{eq:MetricAnaCont} must not be mistaken for a Wick rotation, which is usually performed on the level of coordinates instead of the metric, i.e. rotating the time coordinate $t\rightarrow i\tau$ into the complex plane. \cite{Gerard:2017apb}
Second, causality, maintained in the Epstein-Glaser scheme or analytic renormalization, is restored in the limit $\e\rightarrow0$.

For the transition to Feynman graphs, we use a local version of Wick's theorem \cite{Hollands:2001fb}
\begin{multline}\label{eq:LocalWickOrdering}
	\CT \left\{ :\Ph_1(x_1): ... :\Ph_n(x_n): :\Ph(x_{n+1}):... :\Ph(x_{n+m}): \right\} \\
	= \sum_{\al_1...\al_{n+m}} \frac{1}{|\al_1|! ... |\al_{n+m}|!} :\Ph^{\al_1}_1(x_1)... \Ph^{\al_n}_n(x_n) \Ph^{\al_{n+1}}(x_{n+1})... \Ph^{\al_{n+m}}(x_{n+m}): \times \\
	\times \prod_{\substack{(i,j)\subset\{\al_1,...,\al_{n+m}\}\\i<j}} (\na_{i_1}...\na_{i_l})_{x_i} (\na_{j_1}...\na_{j_k})_{x_j} H_F(x_i, x_j),
\end{multline} 
where the multi-indices $\al_i$ denote functional derivatives with respect to $\ph$ acting on the monomial $\Ph(x_j)$ and $(\na_{i_1}...\na_{i_l})_{x_i}$ denote the symmetrized covariant derivatives from \eqref{eq:CurvPoly}.
Then the quantity under consideration in the extension problem, i.e.
\begin{align}
	\prod_{\substack{(i,j)\subset\{\al_1,...,\al_{n+m}\}\\i<j}} (\na_{i_1}...\na_{i_l})_{x_i} (\na_{j_1}...\na_{j_k})_{x_j} H_F(x_i, x_j) \, ,
\end{align}
can be written in terms of weighted Feynman graphs, i.e. multigraphs $\G(V,E)$ with vertex set $V$ and edge set $E$. 
Assigning a direction to each edge $e\in E$, we get $\pa e=(s(e),t(e))$ with $s,t:E\rightarrow V$, so that the weight over $e$ can be written as
\begin{align}
	u_0[e] \bydef (\na_{i_1}...\na_{i_l})_{x_{s(e)}} (\na_{j_1}...\na_{j_k})_{x_{t(e)}} H_F(x_{s(e)}, x_{t(e)}) \in \CD'(\Om \times \Om \setminus \mbox{``diagonal''}) \, .
\end{align}
For completeness, we assign to each vertex $v_i\in V(\G)$ the weight
\begin{align}
	u_0[v_1,..., v_{n+m}] \bydef :\Ph^{\al_1}_1(x_{v_1})... \Ph^{\al_n}_n(x_{v_n}) \Ph^{\al_{n+1}}(x_{v_{n+1}})... \Ph^{\al_{n+m}}(x_{v_{n+m}}):,
\end{align}
which is smooth in each variable.
We observe that the product over all edges $e\in E(\G)$ is well-defined for $\e > 0$ except for configurations with coinciding vertices due to \eqref{eq:WFAreduction}.
We call the latter configurations the large graph diagonal and denote those by $\oo$, so that
\begin{align}
	u_0^\e[\G] \bydef u_0^\e[v_1,..., v_{n+m}] \prod_{e\in E(\g)} u_0^\e[e] \in \CD'(\Om^{|V(\G)|}\setminus \oo).  
\end{align}  
Furthermore, the extension problem can be reformulated as a problem of local integrability of $u_0^\e[\G]$, which, for a scalar field of dimension 1 in a four-dimensional spacetime, is assessed by the UV-degree of divergence
\begin{align}
	\uvd(u_0^\e[\g]) \bydef 2 |E(\g)| - 4(|V(\g)|-1)
\end{align}
for all $\g\subseteq\G$.
If the UV-degree of divergence is non-negative for at least one $\g\subseteq\G$, then $\G$ requires renormalization and $\g$ is called renormalization part.  
For any graph $\g\subseteq\G$ with $\uvd(u_0^\e[\G])\geq 0$, we have to prescribe a regularization, which is based on the Hadamard regularization, i.e. we want to subtract a Taylor polynomial
\begin{align}\label{eq:NaiveTaylor}
  		t^{d(\g)}_{x_\g|\oli{x}_\g} u_0^\e[\G] \bydef \sum_{|\al|=0}^{d}\limits \frac{(x-\oli{x})^\al}{\al!} D^\al_{x_\g|\oli{x}_\g} u_0^\e[\G]
\end{align}
with sufficiently high order $d(\g)\bydef \lfloor \uvd(u_0^\e[\g]) \rfloor$, where $x_\g$ are the loci of the vertices $V(\g)$ and the subtraction is performed about a point
\begin{align}
	\oli{x}_\g \bydef \frac{1}{2|E(\g)|} \sum_{v\in V(\g)} |E(\g|v)| x_{v},
\end{align} 
with $E(\g|v)$ denoting the set of all edges in $\g$ which are incident to the vertex $v\in V(\g)$.
Since $\oli{x}_\g$ is the singular point of $u_0^\e[\g]$, the Taylor operation \eqref{eq:NaiveTaylor} is well-defined only on the line complement $\G\lineco\g = E(\G)\setminus E(\g)$, so that we define
\begin{align}
	t(\g) \bydef t^{d(\g)}_{x_\g|\oli{x}_\g} \IP(\g)
\end{align}
which acts according to (see Figure \ref{fig:TaylorOperator})
\begin{align}
	t(\g) u_0^\e[\G] \bydef t^{d(\g)}_{x_\g|\oli{x}_\g} \IP(\g) u_0^\e[\G] = u_0^\e[\g] t^{d(\g)}_{x_\g|\oli{x}_\g} u_0^\e[\G\lineco\g].
\end{align}
\begin{figure}[h]
	\centering
	\begin{subfigure}[b]{0.45\textwidth}
		\includegraphics[width=\textwidth]{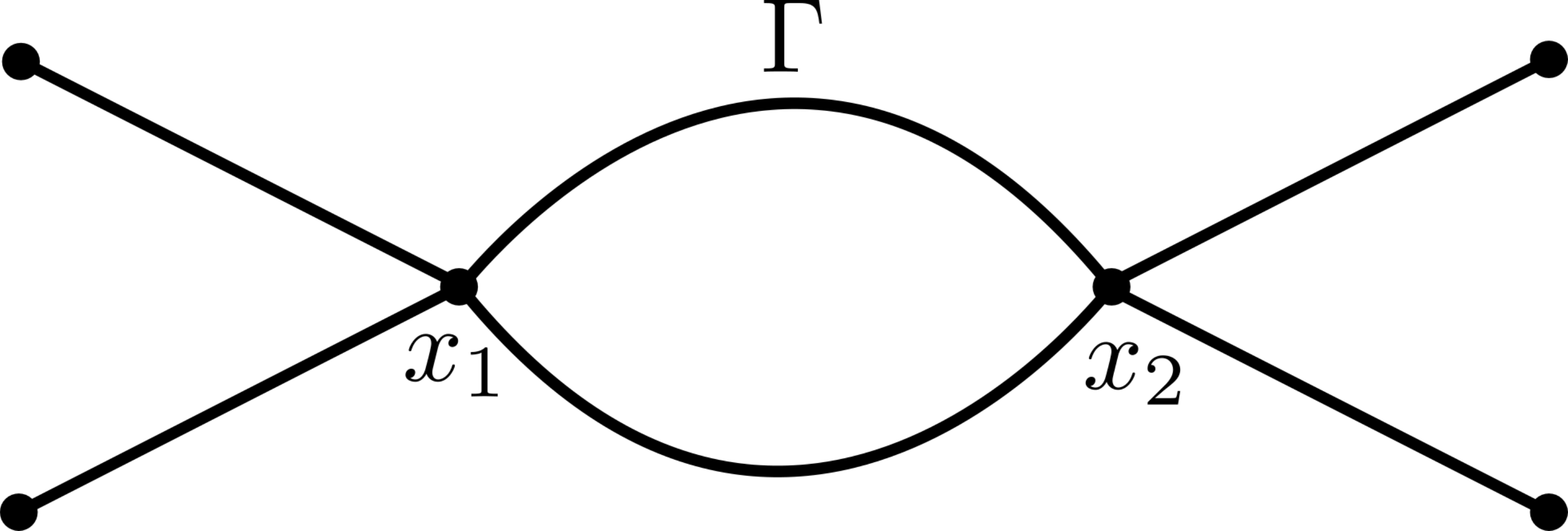}
		\caption{The initial graph $\G$.}
		\label{fig:Gamma}		
	\end{subfigure}
	\begin{subfigure}[b]{0.45\textwidth}
    \centering
		\includegraphics[width=0.5\textwidth]{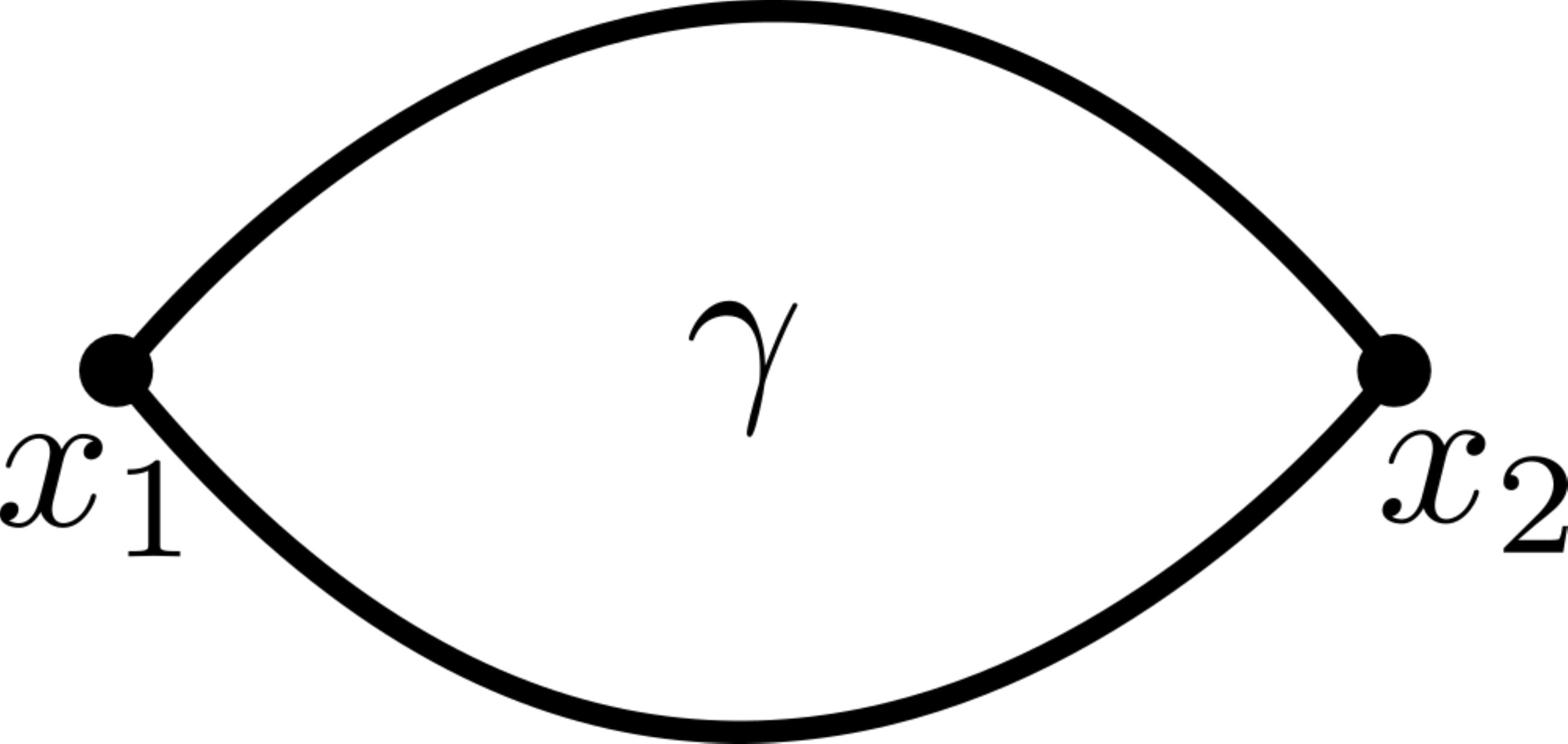}
		\caption{The renormalization part $\g$.}
		\label{fig:gamma}
	\end{subfigure}

	\begin{subfigure}[b]{0.95\textwidth}
		\includegraphics[width=\textwidth]{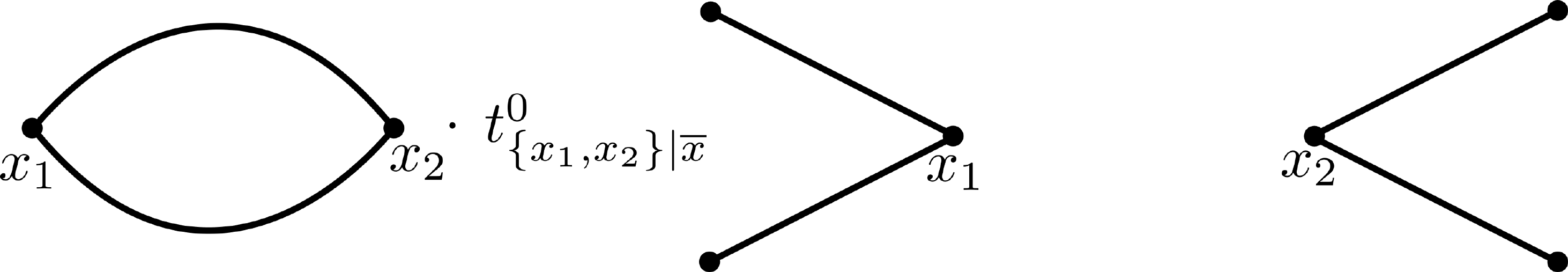}
		\caption{The application of $\IP(\g)$.}
		\label{fig:pretaylor}		
	\end{subfigure}

	\begin{subfigure}[b]{0.8\textwidth}
		\includegraphics[width=\textwidth]{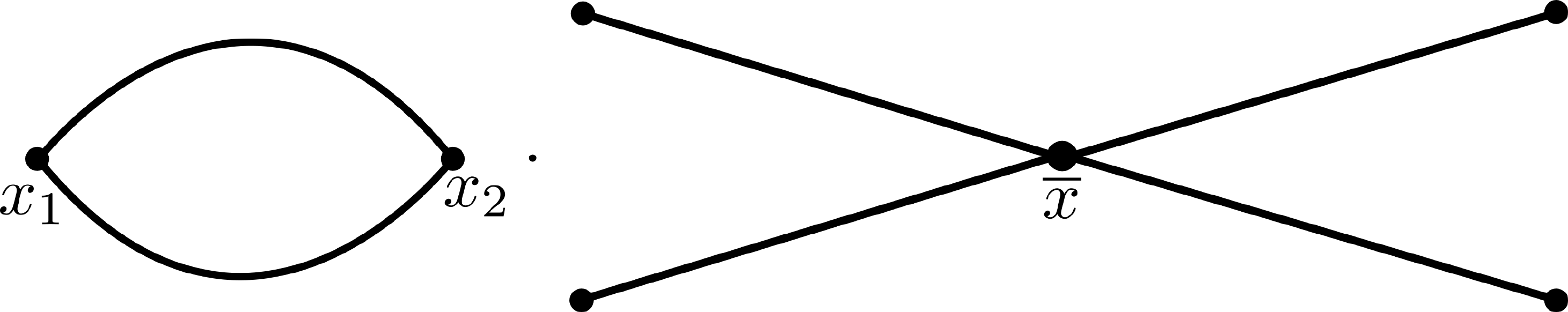}
		\caption{The application of the Taylor operator.}
		\label{fig:taylor}		
	\end{subfigure}
	\caption{Illustration of the action of the Taylor operator $t(\g)$ on a graph $\G$ with one renormalization part $\g$.}
	\label{fig:TaylorOperator}
\end{figure}
It turns out that the regularization shows the desired effect only if the graphs in question are subsumed in partially ordered set, called forests $F$, where the ordering is performed with respect to the the vertex sets $V(\g)$, $\g\subseteq\G$, and the usual set inclusion $\subseteq$. 
In particular, we take only those graphs $\g\subseteq\G$ into account, which are full vertex parts, i.e. the vertex set $V(\g)$ and all edges connecting these vertices.
With this, Bogoliubov's $R$-operation can be written in form of the well-known forest formula
\begin{align}\label{eq:forestformula}
    	Ru_0^\e[\G] \bydef \sum_{F\in\SF}\limits \prod_{\g\in F} (- t(\g)) u_0^\e[\G], 
\end{align} 
where $\SF$ is the set of all forests $F$. 
The Taylor operators are ordered in the sense that $t(\g)$ appears left of $t(\g')$ if $\g\supset \g'$, and no order is preferred if $\g\cap\g'=\emptyset$. 
Then it can be shown that $Ru_0^\e[\G]$ admits a unique extension $Ru^\e[\G]$, which converges to a well-defined distribution $Ru[\G]$ in the limit $\e\rightarrow 0$, so that the naive time-ordering \eqref{eq:NaiveTimeOrdering} and the $R$-operation \eqref{eq:forestformula} define a renormalization scheme \cite{Pottel:2017bb}, i.e. a redefinition of the time-ordered product.
The latter follows from the fact that each graph $\G$ in the expansion \eqref{eq:LocalWickOrdering} is a well-defined distribution in $\CD'(\Om^{|V(\G)|})$ and, thus, the finite sum over all graphs is well-defined as well.
Accordingly, we write
\begin{align}
	\CT_R \left\{ :\Ph_1(f_1): ... :\Ph_n(f_n): :\Ph(g_1):... :\Ph(g_m): \right\} = \sum_{\G\in\mathscr{G}} \langle R_\G u[\G] , f_1 \otimes ... \otimes f_n \otimes g_1 \otimes ... \otimes g_m \rangle \, ,
\end{align}
where $\mathscr{G}$ is the set of all graphs and $R_\G$ indicates that the $R$-operation is performed with respect to the graph $\G$.

In the subsequent sections we want to study products
\begin{align}\label{eq:naive_OPE_Product}
	:\Ph_1:(x_1)... :\Ph_n:(x_n) \, ,
\end{align}
with $:\Ph_j: \in \SCB(M,g)$, in the limit of coinciding arguments $x_j\rightarrow x$. 
Clearly, Wick ordering renders the product \eqref{eq:naive_OPE_Product} well-defined. 
But inserting \eqref{eq:naive_OPE_Product} into a time-ordered product and applying Wick's theorem \eqref{eq:LocalWickOrdering}, we may ask which additional divergences occur in the graph weight if we perform the limit of coinciding arguments. 
We observe that the graphs $\G$ are not necessarily connected so that we write
\begin{align}
	\G = \G_1 \cup ... \cup \G_c,
\end{align}
where all $\G_k$ are connected. 
Those graphs $\G_k$ may either contain $x_j$-vertices or $x_{n+i}$-vertices or both. 
Only the latter case is interesting for our purposes since graphs with only $x_{n+i}$-vertices do not contribute in the limit of coinciding arguments while graphs with only $x_j$-vertices are regularized by Wick ordering. 
However graphs with only $x_j$-vertices have to be taken into account if there exists at least one graph with both types of vertices in $\G$. 
The reason for this lies in the fact that graphs may join in the limit and new renormalization parts appear. 
Let us first look at a connected graph $\G_k$ with $x_j$- and $x_{n+i}$-vertices. 
In contrast to the momentum space scheme, there may exist renormalization parts in $\G_k$, which contain one or more $x_j$-vertices. 
Therefore three effects may occur in the limit of coinciding arguments. 
First, the limit introduces overlapping renormalization parts due to joining them in one vertex. 
Second, a subgraph $\g\subseteq\G_k$ increases the scaling degree, for which, denoting the graph with joined vertex by $\tilde \g$, we have
\begin{align}
	\uvd(u^\e_0[\g]) < \uvd(u^\e_0[\tilde \g]). 
\end{align}
For completeness, we introduce the splitting of vertices (the inverse of joining vertices), which we acts according to $\hat{\tilde\g} = \g$. 
Third, the limit introduces entirely new renormalization parts. 
Introducing (see Figure \ref{fig:gammatodelta})
\begin{align}
	\De \bydef \tilde\G \qquad \& \qquad \G \bydef \hat\De,
\end{align}
\begin{figure}[h]
	\centering
	\includegraphics[width=0.8\textwidth]{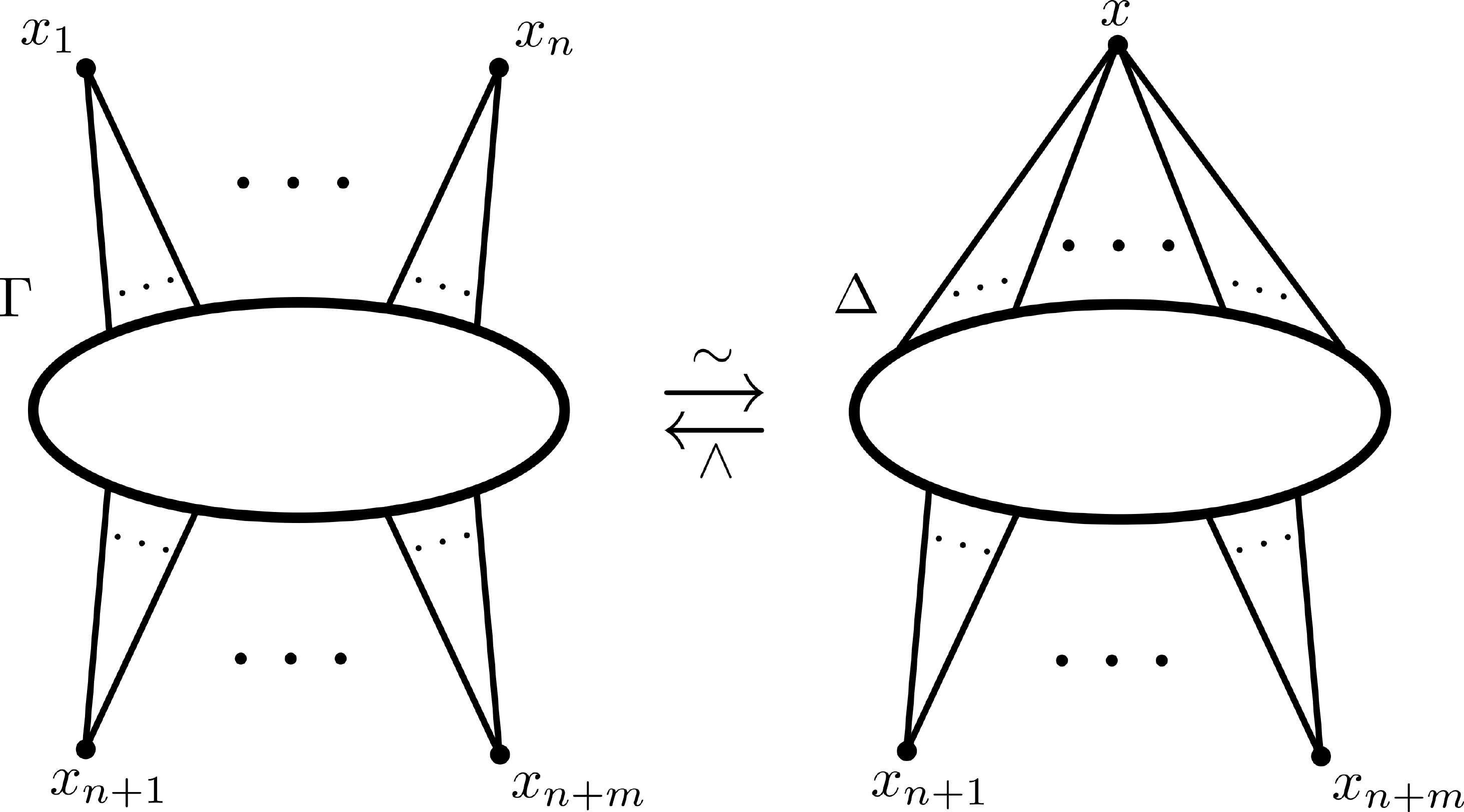}
	\caption{Illustration of the operations $\tilde \G$ and $\hat \De$. }
	\label{fig:gammatodelta}
\end{figure}
we claim that, due to the properties of the $R$-operation, effects of the limit can be traced back and controlled, i.e. 
\begin{align}
	R_\De u_0^\e[\G] & = R_\G u_0^\e[\G] + X u_0^\e[\G] \\
	& = R_\G u_0^\e[\G] + X_\G u_0^\e[\G] + X_\De u_0^\e[\G],
\end{align}
where the indices $\G$ and $\De$ denote the forest formula taken with respect to renormalization parts in $\G$ and $\De$, respectively. 
We distinguish corrections $X_\G$ which stem from overlapping renormalization parts in $\De$ and corrections $X_\De$ which stem from new renormalization parts or renormalization parts with increased scaling degree after the coincidence limit is performed.
Specifically, in the limit $x_j \rightarrow x$, $j\in \{1,...,n\}$, we find $\G\rightarrow\De$ and, thus,
\begin{align}
	R_\De u_0^\e[\De] = \lim_{x_j\rightarrow x} R_\De u_0^\e[\G] = R_\G u_0^\e[\De] + X_\G u_0^\e[\De] + X_\De u_0^\e[\De] \, ,
\end{align}
which is a well-defined expression according to \eqref{eq:forestformula}. 
However, we remark that the coincidence limit does not modify the $R$-operation and hence the two operations do not commute in general.
Therefore it is the objective of the subsequent sections to derive a definition of normal products
\begin{align}
	N_\de[\Ph_1(f_1)... \Ph_n(f_n)]
\end{align}
with degree $\de$ such that the coincidence limit 
\begin{align}\label{eq:NormalProductHypothesis}
	\lim_{f_j\rightarrow f} \CT_R \{ N_\de[\Ph_1(f_1)... \Ph_n(f_n)] :\Ph(g_{1}):... :\Ph(g_{m}): \}
\end{align}
results in a well-defined distribution, where the limit has to be understood in the sense that
\begin{align}
	f_1 \otimes ... \otimes f_n \rightarrow f \, .
\end{align}
Certainly, the definition will hold only for a single but arbitrary correlation function.
For an example of a normal product defined for a perturbative treatment of a theory with one scalar field and quartic self-interaction, we refer to \cite{Pottel:2017fwz}.

The possibility of finding renormalization parts with differing subtraction degree is of more general nature. 
Indeed, we observe that the $R$-operation still defines a renormalization method if we choose subtraction degrees $\de(\g)$ with $\de(\g)\geq d(\g)$ for some renormalization part $\g\subseteq\G$. 
It turns out that those choices are not completely arbitrary, but, more importantly, different choices can be related to each other by Zimmermann identities \cite{Zimmermann:1972te}. 
As indicated above, those identities are indispensable for the definition of normal products.

In the following we work with the analytic continuation of the metric $g^\e$, $\e>0$, and the non-extended distributions $u_0^\e[\G]$. 
The existence of the extended distribution $Ru[\G]$ is just a consequence of the BPHZ renormalization scheme. \cite{Pottel:2017fwz}

\section{Zimmermann Identity}

The possibility of oversubtraction admits almost arbitrary choices of rendering the time-ordered products
\begin{align}
	\CT^{(\mathrm{conn})} \left\{ \prod_{j=1}^n \Ph_j(f_j) \right\},
\end{align}
expanded in (connected) graphs $\G$, extendible, where we find two options. 
Namely, we may either assign degrees 
\begin{align}
	\de(\g) > d(\g),
\end{align}
to renormalization parts $\g\subset \G$, or assign degrees to the constituting Wick monomials, which is denoted by 
\begin{align}\label{eq:NormalProductMonomial}
	N_{\de_j}[\Ph_j(f_j)]
\end{align}
for $\de_j \geq \dim(\Ph_j)$, where $\dim(\Ph_j)$ is the dimension of elementary field operators as well as their covariant derivations, hence not the engineering dimension in general. 
We observe that \eqref{eq:NormalProductMonomial} is not in conflict with our working hypothesis \eqref{eq:NormalProductHypothesis} since it contains only one field monomial.
\begin{rmk}
	It is worth noting that the $R$-operation admits writing monomials $\Ph_j$ without explicitly performing the Wick ordering. 
	Namely, if there exists an edge $e\in E(\G)$ with $s(e)=t(e)$, then the edge weight is a renormalization part with one element in the vertex set. 
	Hence the subtraction point of the Taylor operator coincides with this vertex, i.e. $s(e)=t(e)=\oli x_e$, so that the line complement $u_0^\e[\G\lineco\g]$ vanishes for each configuration of $e$, i.e. for each $s(e)$.
\end{rmk}
We refer to the former choice of subtraction degrees as anisotropic and to the latter as isotropic \cite{Clark:1976ym}. 
In the following, we work with isotropic degrees, which can be translated into degrees for renormalization parts. 
For a Feynman graph $\g$, the degree of divergence is given by
\begin{align}
	\uvd(\g) = (\dim(M) - 2 \dim(\ph))|E(\g)| - \dim(M)(|V(\g)|-1) + |\SD_\g|,
\end{align}
where $\SD_\g$ is the set of covariant derivatives acting on fields contributing to $\g$. The number of edges $|E(\g)|$ is exactly half of the number of elementary fields $|\SE_\g|$ contributing to $\g$ and we define $V(\g) \bydef \SV_\g$ so that for $\dim(M)=4$ and $\dim(\ph)=1$ we obtain
\begin{align}
	\uvd(\g) = 4 + |\SE_\g| - 4|\SV_\g| + |\SD_\g| .
\end{align}
The sets $\SE_\g$ and $\SD_\g$ may be rewritten in terms of monomials $\Ph$ assigned to vertices of $\g$, i.e. 
\begin{align}
	|\SE_\g| + |\SD_\g| = \sum_{v\in\SV_\g} \dim(\Ph_v) - |\oli\SE_\g| - |\oli\SD_\g|,
\end{align}
where $\oli\SE_\g$ denotes the set of elementary fields contributing to external lines and $\oli\SD_\g$ denotes the set of covariant derivatives at fields contributing to external lines. Introducing the degrees of oversubtraction $\de_j \geq \dim(\Ph_j)$ and defining the \emph{codegree} of $\g$ to be $\oli d(\g) \bydef |\oli\SE_\g| + |\oli\SD_\g|$, the degree of divergence for oversubtractions is given by
\begin{align}
	\de(\g) = 4 + \sum_{v\in\SV_\g} (\de_v - 4) - \oli d(\g).
\end{align}
For mainly didactical purposes, we begin the derivation of relations among different choices of
\begin{align}
	\CT_{R}\left\{ \prod_{j=1}^n N_{\de_j}[\Ph_j(f_j)] \right\}
\end{align}
in the case of one connected graph $\De$ and a monomial $\Ph(f)$ to which we assign the vertex $V_0\in V(\De)$ and two degrees 
\begin{align}
	a>b\geq \dim(\Ph).
\end{align}
This is the analogue to the treatment in \cite{Zimmermann:1972te} but restricted to an intermediate step, since we neither consider the full perturbative expansion in a power series of a theory nor do we have reduction techniques in the sense of LSZ at our disposal. Also the sum over all contributions is left for a later stage of this section. 
\begin{lem}\label{le:ZI}
	Let $u_0^\e[\De]$ be the weight over a connected graph $\De$, which contains a distinguished vertex $V_0$ with $a>b\geq \dim(\Ph)$ assigned. The difference of the prescriptions under the $R$-operation is given by
	\begin{align}
		R^{(a)}_\De u_0^\e[\De] - R^{(b)}_\De u_0^\e[\De] = - \sum_{\ta\in\ST} \sum_{|\al|=b+1}^a \frac{1}{\al!} G_{\oli x_\ta,\al}[\ta] R^{(a)}_{\De/\ta} (D^\al_{\oli x_\ta} u_0^\e[\De/\ta]),
	\end{align}
	where $\ST$ is the set of renormalization parts containing $V_0$ and 
	\begin{align}
		G_{\oli x_\ta,\al}[\ta] \bydef \int (x-\oli x_\ta)^\al R^{(b)}_{\ta^\perp} u_0^\e[\ta](x) f(x) d\m_x,
	\end{align}
	where $f\in\CD(M^{|V(\ta)|-1})$ and $\ta^\perp$ denotes the set of normal $\ta$-forest, i.e. the set of all forests over $\ta$ which do not contain $\ta$.
\end{lem}
\begin{proof}
We have to compare $R_\De^{(a)}u_0^\e[\De]$ to $R_\De^{(b)}u_0^\e[\De]$ and decompose the set of all $\De$-forest $\SF(\De)$ into the set of $\De$-forests $\SF_1(\De)$, where no element is containing a graph with vertex $V_0$, and the complementary set of $\De$-forests $\SF_0(\De)$, i.e.
\begin{align}
 	\SF(\De) \bydef \SF_1(\De) \cup \SF_0(\De).
\end{align} 
Indicating with $t^{(a)}(\g)$ and $t^{(b)}(\g)$ two choices of oversubtractions, respectively, we write
\begin{align}
	R^{(a)}_\De u_0^\e[\De] & = \sum_{F\in\SF(\De)} \prod_{\g\in F} (-t^{(a)}(\g)) u_0^\e[\De] \\
	& = \sum_{F\in\SF_1(\De)} \prod_{\g\in F} (-t(\g)) u_0^\e[\De] + \sum_{F\in\SF_0(\De)} \prod_{\g\in F} (-t^{(a)}(\g)) u_0^\e[\De] .
\end{align}
Since for any $F\in \SF_0(\De)$ there exists a $\g\in F$ containing $V_0$, thus $\g$ is subtracted by $t^{(a)}(\g)$, which can be rewritten as
\begin{align}\label{eq:ZITaylorDecomp}
	t^{(a)}(\g) = t^{(b)}(\g) + (t^{(a)}(\g) - t^{(b)}(\g)).
\end{align}
Expanding the previous sum \eqref{eq:ZITaylorDecomp} for all graphs containing $V_0$ leads to new forests in $\SF_0(\De)$, some of which consisting of only factors with degree $b$. We denote these forests by $\SF_b$ and find
\begin{align}
	R^{(b)}_\De u_0^\e[\De] = \sum_{F\in\SF_1(\De)} \prod_{\g\in F} (-t(\g)) u_0^\e[\De] + \sum_{F\in\SF_b(\De)} \prod_{\g\in F} (-t^{(b)}(\g)) u_0^\e[\De]. 
\end{align}
Then we compute the difference 
\begin{align}
	R^{(a)}_\De u_0^\e[\De] - R^{(b)}_\De u_0^\e[\De] = \sum_{F\in\SF_0(\De)\setminus\SF_b(\De)} \prod_{\g\in F} (-t^{(a)}(\g)) u_0^\e[\De] \bydef Xu_0^\e[\De].
\end{align}
By construction there are only forests contributing to $Xu_0^\e[\De]$ which contain at least one graph $\g\subseteq\De$ subtracted by $t^{(a)}(\g)-t^{(b)}(\g)$. 
Among those graphs, we choose the minimal graph $\ta\in F$ for some forest $F\in\SF_0(\De)\setminus\SF_b(\De)$. 
Then all graphs $\g'\subset \ta$, $\g'\in F$, are either subtracted by $t^{(b)}(\g')$ or $t(\g)$ and all graphs $\g\supset\ta$ are either subtracted by $t^{(b)}(\g)$ or $t^{(a)}(\g)-t^{(b)}(\g)$. 
But we may add up the latter to subtractions $t^{(a)}(\g)$. 
Furthermore every $\g\subseteq\De$ containing $V_0$ is at least once minimal and has $t^{(a)}(\g)-t^{(b)}(\g)$ assigned to it. 
Therefore we subsume those in a set $\ST$. For each $\ta\in\ST$, we know that graphs $\g\supset\ta$ are subtracted with degree $a$ and all graphs $\g'\subset\ta$ are subtracted with degree $b$. 
With respect to any $\ta$, we construct all superforests $\oli F_\ta$, i.e. all posets of renormalization parts $\g$ fulfilling $\g\supset\ta$, and all subforests $\ul F_\ta$, i.e. all posets of renormalization parts $\g'$ fulfilling $\g'\subset\ta$. 
Those forests may be subsumed in sets $\oli \SF_\ta$ and $\ul \SF_\ta$, respectively. 
With this we have
\begin{align}
	Xu_0^\e[\De] = \sum_{\ta\in\ST} \sum_{\oli F_\ta \in \oli \SF_\ta} \sum_{\ul F_\ta \in \ul \SF_\ta} \prod_{\g\in\oli F_\ta} (-t^{(a)}(\g)) (-(t^{(a)}(\ta)-t^{(b)}(\ta))) \prod_{\g'\in \ul F_\ta} (-t^{(b)}(\g')) u_0^\e[\De].
\end{align}
Spelling out the Taylor operator
\begin{align}
	t^{(a)}(\ta) - t^{(b)}(\ta) = \sum_{|\al|=b+1}^a \frac{(x-\oli x_\ta)^\al}{\al!} D^\al_{x_\ta | \oli x_\ta} \IP(\ta)
\end{align}
and writing
\begin{align}\label{eq:NoTaylorOutside}
	\sum_{\ul F_\ta \in \ul \SF_\ta} \prod_{\g'\in \ul F_\ta} (-t^{(b)}(\g')) u_0^\e[\De] = u_0^\e[\De\setminus\ta] R^{(b)}_{\ta^\perp} u_0^\e[\ta],
\end{align}
with $\ta^\perp$ denoting the set of all normal $\ta$-forests, we obtain
\begin{multline}
	-(t^{(a)}(\ta) - t^{(b)}(\ta)) \sum_{\ul F_\ta \in \ul \SF_\ta} \prod_{\g'\in \ul F_\ta} (-t^{(b)}(\g')) u_0^\e[\De] \\
	= - \sum_{|\al|=b+1}^a \frac{1}{\al!} (D^\al_{x_\ta | \oli x_\ta} u_0^\e[\De\setminus\ta]) (x-\oli x_\ta)^\al R^{(b)}_{\ta^\perp} u_0^\e[\ta].
\end{multline}
The argument for \eqref{eq:NoTaylorOutside} can be read off from
\begin{align}
	D^\al_{x_\ta | \oli x_\ta} \frac{(x-\oli x_{\g'})^\be}{\be!} D^\be_{x_{\g'}|\oli x_{\g'}} = 
	\begin{cases}
		0 & \mbox{for } \be \nsubseteq \al \\
		D^\al_{x_\ta | \oli x_\ta} & \mbox{otherwise}
	\end{cases}
\end{align}
when $\g'\subset\ta$ and at $V(\g')\cap V(\ta) \cap V(\De\lineco\ta)$. 
Note further that $(x-\oli x_\ta)^\al R^{(b)}_{\ta^\perp} u_0^\e[\ta]$ is locally integrable for every $\al$ since both $R^{(a)}_\De u_0^\e[\De]$ and $R^{(b)}_\De u_0^\e[\De]$ are locally integrable. 
Hence after integrating over all but one argument in $u_0^\e[\ta]$, we obtain a function $G_{\oli V,\al}[\ta]$, which is multiplied to
\begin{align}
	- \sum_{\oli F_\ta \in \oli \SF_\ta} \prod_{\g\in\oli F_\ta} (-t^{(a)}(\g)) \sum_{|\al|=b+1}^a \frac{1}{\al!} (D^\al_{x_\ta | \oli x_\ta} u_0^\e[\De\setminus\ta]).
\end{align}
Indeed, setting all vertices $V(\ta)\cap V(\De\lineco\ta)$ to $\oli x_\ta$ in the Taylor polynomial corresponds to contracting $\ta$ to a point, i.e. 
\begin{align}
	\De\setminus\ta \mapsto \De/\ta.
\end{align}
Hence we obtain
\begin{align}
	D^\al_{x_\ta | \oli x_\ta} u_0^\e[\De\setminus\ta] = D^\al_{\oli x_\ta} u_0^\e[\De/\ta],
\end{align}
which fits nicely to our observation above. 
It follows from the construction of the $R$-operation that the vertex at $\oli x_\ta$ is determined up to counterterms of degree $|\al|$ and no renormalization parts which are subgraphs to $\ta$ may take part in further subtractions. 
Furthermore the degrees of $\g\supset\ta$ remain the same and we may view $\oli F_\ta$ as a forest over $\De/\ta$ instead of $\De$ such that
\begin{align}
	\sum_{\oli F_\ta \in \oli \SF_\ta} \prod_{\g\in\oli F_\ta} (-t^{(a)}(\g)) = R^{(a)}_{\De/\ta}
\end{align}
and we conclude
\begin{align}
	Xu_0^\e[\De] = - \sum_{\ta\in\ST} \sum_{|\al|=b+1}^a \frac{1}{\al!} R^{(a)}_{\De/\ta} (D^\al_{\oli x_\ta} u_0^\e[\De/\ta]) G_{\oli x_\ta,\al}[\ta].
\end{align}
\end{proof}
We denote the resulting vertex after the contraction of $\ta$ by $\oli V \in V(\De / \ta)$ and observe that $\oli V$ behaves like a vertex in some part of a "new" correlation function. 
This leads to constraints on the number of derivatives depending on the number of incident lines in $\oli V$ after summing over all contributing graphs. 
Before performing the summation, we generalize Lemma \ref{le:ZI} to the case of two choices of subtraction degrees $\de^{(1)}$ and $\de^{(2)}$ for a connected weighted graph $\De.$ 
Both choices have to be sufficiently high but we do not require them to be strictly ordered and set 
\begin{align}
	\sum_{\de^{(2)}}^{\de^{(1)}} = - \sum_{\de^{(1)}}^{\de^{(2)}}
\end{align}
as convention whenever $\de^{(2)} > \de^{(1)}$.
\begin{prop}\label{pr:GZI}
	Let $u_0^\e[\De]$ be the weight over a connected Feynman graph $\De$. 
	For two choices $\de^{(1)}(\g)$ and $\de^{(2)}(\g)$ with respect to any $\g\subseteq\De$, the difference of the prescriptions under the $R$-operation is given by 
	\begin{multline}
		R^{(1)}_\De u_0^\e[\De] - R^{(2)}_\De u_0^\e[\De] = 
		\sum_{\{\ta_j\}}  (-1)^{|\{\ta_j\}|} \prod_j \sum_{|\al_j|=\de^{(2)}_{\ta_j}+1}^{\de^{(1)}_{\ta_j}} \frac{1}{\al_j!} R^{(1)}_{\De/\{\ta_j\}} D^\al_{\oli x_{\ta_j}} u_0^\e[\De / \{\ta_j\}] G_{\oli x_{\ta_j}, \al_j}[\ta_j],
	\end{multline}
	where $\{\ta_j\}$ is a set of mutually disjoint renormalization parts.
\end{prop}
\begin{proof}
	In analogy to the proof of Lemma \ref{le:ZI} and \cite{Clark:1976ym}, we write
\begin{align}
	R^{(1)}_\De u_0^\e[\De] & = \sum_{F\in\SF} \prod_{\g\in F} (-t_1(\g)) u_0^\e[\De] \\
	& = \sum_{F\in\SF} \prod_{\g\in F} (-t_2(\g) - (t_1(\g)-t_2(\g))) u_0^\e[\De],
\end{align}
where we used
\begin{align}\label{eq:ExpansionChoiceOneTwo}
	t_1(\g) = t_2(\g) + (t_1(\g) - t_2(\g)).
\end{align}
Expanding the forest formula in $t_2(\g)$ and $(t_1(\g)-t_2(\g))$, there is always exactly one contribution per forest $F\in\SF$ which selects only operators $t_2(\g)$. 
Summing those contributions up results in $R^{(2)}_\De u_0^\e[\De]$ so that we may focus on the terms separating the two choices 
\begin{align}
	R^{(1)}_\De u_0^\e[\De] - R^{(2)}_\De u_0^\e[\De] \bydef X_{12}u_0^\e[\De] .
\end{align}
In any forest $F$ of $X_{12}u_0^\e[\De]$, we can find at least one minimal renormalization part $\ta_j$ being assigned to the difference of the Taylor operators, since there may exist mutually disjoint, minimal renormalization parts. 
Due to the expansion of \eqref{eq:ExpansionChoiceOneTwo}, every element of $F$ is at least once minimal in that sense. 
Therefore we sum over all possible subsets $\{\ta_j\}\subset F$ of mutually disjoint, Taylor-difference minimal $\ta_j$. 
Using the notation from above, we obtain
\begin{align}
	X_{12}u_0^\e[\De] = \sum_{F\in F_{12}} \sum_{\{\ta_j\} \subset F} \prod_{\g\in \oli F_{\{\ta_j\}}} (-t_1(\g)) \prod_j (-(t_1(\ta_j) - t_2(\ta_j))) \prod_{\g'\in \ul F_{\{\ta_j\}}} (-t_2(\g')) u_0^\e[\De].
\end{align}
We may split
\begin{align}
	u_0^\e[\De] = u_0^\e[\De\setminus\{\ta_j\}] \prod_j u_0^\e[\ta_j]
\end{align}
and interchange the sum over all forests with the sum over $\{\ta_j\}$
\begin{align}
	\sum_{F\in\SF} \sum_{\{\ta_j\}} = \sum_{\{\ta_j\}} \sum_{\oli F_{\{\ta_j\}} \in \oli\SF_{\{\ta_j\}}}  \sum_{\ul F_{\{\ta_j\}} \in \ul\SF_{\{\ta_j\}}} 
\end{align}
using that $\{\ta_j\}$ can be any set of mutually disjoint renormalization parts in $\De$ so that, by the same argument as above, we abbreviate
\begin{align}
	\sum_{\ul F_{\{\ta_j\}} \in \ul\SF_{\{\ta_j\}}} \prod_{\g'\in \ul F_{\{\ta_j\}}} (-t_2(\g')) u_0^\e[\De] = u_0^\e[\De\setminus \{\ta_j\}] \prod_j R^{(2)}_{\ta^\perp_j} u_0^\e[\ta_j].
\end{align}
Further, we spell out
\begin{align}
	\prod_j (-(t_1(\ta_j) - t_2(\ta_j))) u_0^\e[\De\setminus \{\ta_j\}] = (-1)^{|\{\ta_j\}|} \prod_j \sum_{|\al_j|=\de^{(2)}_{\ta_j}+1}^{\de^{(1)}_{\ta_j}} \frac{(x - \oli x_{\ta_j})^{\al_j}}{\al_j!} D^\al_{x_{\ta_j}|\oli x_{\ta_j}} u_0^\e[\De\setminus \{\ta_j\}],
\end{align}
where again
\begin{align}
	\prod_j D^\al_{x_{\ta_j}|\oli x_{\ta_j}} u_0^\e[\De\setminus \{\ta_j\}] = \prod_j D^\al_{\oli x_{\ta_j}} u_0^\e[\De / \{\ta_j\}].
\end{align}
In a last step, we observe that due to the properties of the $R$-operation
\begin{align}
	\sum_{\oli F_{\{\ta_j\}} \in \oli\SF_{\{\ta_j\}}} \prod_{\g\in \oli F_{\{\ta_j\}}} (-t_1(\g)) u_0^\e[\De / \{\ta_j\}] & = \sum_{F\in \SF(\De/\{\ta_j\})} \prod_{\oli \g\in F} (-t_1(\oli \g)) u_0^\e[\De / \{\ta_j\}] \\
	& = R^{(1)}_{\De/\{\ta_j\}} u_0^\e[\De / \{\ta_j\}]
\end{align}
and conclude
\begin{align}\label{eq:GZIGraph}
	X_{12}u_0^\e[\De] = \sum_{\{\ta_j\}}  (-1)^{|\{\ta_j\}|} \prod_j \sum_{|\al_j|=\de^{(2)}_{\ta_j}+1}^{\de^{(1)}_{\ta_j}} \frac{1}{\al_j!} R^{(1)}_{\De/\{\ta_j\}} D^\al_{\oli x_{\ta_j}} u_0^\e[\De / \{\ta_j\}] (x - \oli x_{\ta_j})^{\al_j} R^{(2)}_{\ta^\perp_j} u_0^\e[\ta_j] .
\end{align}
Integrating out all but one variables in each $(x - \oli x_{\ta_j})^{\al_j} R^{(2)}_{\ta^\perp_j} u_0^\e[\ta_j]$, we arrive at the assertion.
\end{proof}
We already indicated above that the vertices $\oli V_j$, resulting from contractions of renormalization parts $\ta_j$, shall be represented by monomials in time-ordered products when summing over all graphs. 
The degree assigned to those monomials is determined by the scheme, but the number of derivatives, remaining from the Taylor operators, is constrained by the dimension of incident lines in $\oli V_j$. 
Namely, the monomial of vertex $\oli V_j$ has to have at least the degree of the number of incident lines $\oli\SE_{\ta_j}$, multiplied by the degree of the elementary field which is trivial in our case, and the number of covariant derivatives $\oli\SD_{\ta_j}$, which are applied to those incident lines. 
Additionally, we have to take the subtraction degree $\max\{\de^{(1)}_{\ta_j},\de^{(2)}_{\ta_j}\}$ into account so that
\begin{align}
	\de_{\oli V_j} & = |\max\{\de^{(1)}_{\ta_j},\de^{(2)}_{\ta_j}\}| + |\oli\SE_{\ta_j}| + |\oli\SD_{\ta_j}|.
\end{align}
In order to conclude this section, we need to return to the full time-ordered product. 
Therefore we consider the sum over all graphs, which we obtained by Wick's theorem. In the spirit of \cite{Clark:1976ym}, we identify for any $\ta_j\in \{\ta_j\}$ its vertex set 
\begin{align}
	V(\ta_j) \bydef \SV_j\subseteq \{1,...,n\},
\end{align}
where $\SV_j\cap \SV_i = \emptyset$ since $\ta_j\cap\ta_i=\emptyset$ for $\ta_j,\ta_i\in\{\ta_j\}$. 
For each element $v\in V(\De)$, we find a certain set of elementary fields being associated to lines of $\De$. This set is independent of derivatives, curvature terms or constants of the theory. 
We denote it by 
\begin{align}
	\{v_1,...,v_k\}
\end{align}
such that we can define the set of all elementary fields 
\begin{align}
	\SE_j \subseteq \{ \{v_{i_{11}},...,v_{i_{1j_1}}\},...,\{v_{i_{n1}},...,v_{i_{nj_n}}\} \}
\end{align}
constructing $\ta_j$. 
Note that the pair $(\SV_j,\SE_j)$ does not determine $\ta_j$ uniquely in general. 
\begin{thm}\label{th:GZI}
	Let $\de^{(1)}$ and $\de^{(2)}$ be two sets of subtraction degrees for a connected time-ordered product 
	\begin{align}
		\CT^\mathrm{conn}_{R^{(j)}} \left\{ \prod_{i=1}^n N_{\de^{(j)}_i}[\Ph_i(f_i)]  \right\}.
	\end{align}
	Their difference is given by 
	\begin{multline}
		\CT^\mathrm{conn}_{R^{(1)}} \left\{ \prod_{i=1}^n N_{\de^{(1)}_i}[\Ph_i(f_i)]  \right\} - \CT^\mathrm{conn}_{R^{(2)}} \left\{ \prod_{i=1}^n N_{\de^{(2)}_i}[\Ph_i(f_i)]  \right\} \\
		= \sum_{c\geq 1} \sum_{\{\SV_i\}_c} \sum_{\{\SE_i\}_c} (-1)^c \sum_{\al} \frac{1}{\al!} \CT^\mathrm{conn}_{R}\left\{ \prod_{k\in \oli{\{\SV_i\}}_c} N_{\de^{(1)}_k}[\Ph_k(f_k)] \prod_{\SV_l\in\{\SV_i\}_c} N_{\de^{(1)}_{\oli V_l}} [\SCL_{l,\al_l} D^{\al_l}_{\oli x_{\ta_l}} \oli{\SE_l} (\oli f_l)] \right\},
	\end{multline} 
	with
	\begin{align}
		\min\{\de^{(1)}_{\oli V_j}, \de^{(2)}_{\oli V_j}\} < |\oli\SE_j| + |\oli\SD_j| + |\al_j| \leq \max\{\de^{(1)}_{\oli V_j}, \de^{(2)}_{\oli V_j}\}
	\end{align}
	and 
	\begin{align}
		\SCL_{i,\al_i} = \CT^\mathrm{conn}_{R} \left\{ \prod_{l'\in \SV_i} (x-\oli x_{\SV_i})^{\al_{i}} N_{\de^{(2)}_{l'}} [\Ph_{l'}/\oli\SE_i(f_l)] \right\},
	\end{align}
	where
	\begin{align}
		\Ph_l/\oli\SE_j \bydef \prod_{(l,l')\in \oli{\SE}_j} \frac{\de}{\de (\na...\na)_{ll'}\ph(y_{ll'})} \Ph_l(x_l)
	\end{align}
	and $(l,l')$ denote the $l$-th vertex in $\SV_j$ and the $l'$-th elementary field $\ph_{ll'}\in\oli\SE_j$ at the $l$-th vertex, which carries covariant derivatives $(\na...\na)_{ll'}$.
\end{thm}
\begin{proof}
Setting $|\{\ta_j\}_c| = c$ and $\De/\{\ta_j\}_c \bydef \oli \De$, we have
\begin{align}
	\sum_\De \sum_{\{\ta_j\}} = \sum_{c\geq 1} \sum_{\{\SV_i\}_c} \sum_{\{\SE_i\}_c} \sum_{\{\ta_i\}_c} \sum_{\oli \De}
\end{align}
and observe that $\oli\De$ is constructed by the complementary vertex set $\oli{\{\SV_j\}}_c$ as well as vertices $\oli V_1,...,\oli V_c$ resulting from contractions of $\{\ta_j\}_c$ and lines among elements in the complementary set of elementary fields $\oli{\{\SE_j\}}_c$. 
With this, $\oli\De$ is independent of the realization $\{\ta_j\}_c$ and we associate to elements in $\oli{\{\SV_j\}}_c$ the monomials
\begin{align}
	N_{\de_i^{(1)}}[\Ph_i(f_i)]
\end{align}
and to vertices $\oli V_j$ the monomials
\begin{align}
	N_{\de_{\oli V_j}^{(1)}}[D^{\al_j} \oli{\SE_j} (\oli f_j)].
\end{align}
In the same manner, we view any $\ta_j$ as independent contribution to
\begin{align}
	\prod_{l\in\SV_j} N_{\de^{(2)}_l} [\Ph_l / \oli \SE_j (f_l)].
\end{align}
Then we obtain for the sum over all connected graphs $\De$ in \eqref{eq:GZIGraph}
\begin{multline}
	\sum_\De X_{12}u_0^\e[\De] \\ 
	= \sum_{c\geq 1} \sum_{\{\SV_i\}_c} \sum_{\{\SE_i\}_c} (-1)^c \sum_{|\al|=\de^{(2)}+1}^{\de^{(1)}} \frac{1}{\al!} \CT^\mathrm{conn}_{R}\left\{ \prod_{k\in \oli{\{\SV_i\}}_c} N_{\de^{(1)}_k}[\Ph_k(f_k)] \prod_{\SV_l\in\{\SV_i\}_c} N_{\de^{(1)}_{\oli V_l}} [D^{\al_l} \oli{\SE_l} (\oli f_l)] \right\} \times \\
	\times \prod_{i=1}^c \underbrace{\CT^\mathrm{conn}_{R} \left\{ \prod_{l'\in \SV_i} (x-\oli x_{\SV_i})^{\al_{i}} N_{\de^{(2)}_{l'}} [\Ph_{l'}/\oli\SE_i(f_l)] \right\} }_{\bydef \SCL_{i,\al_i} (\oli f_i)},
\end{multline}
where $\SCL_{i,\al_i} (\oli f_i)$ is a sufficiently regular function in the argument $\oli x_{\SV_i}$, which we may assign to the monomial $N_{\de^{(1)}_{\oli V_i}} [D^{\al_i} \oli{\SE_i} (\oli f_i)]$. In order to conclude the proof of the theorem, we observe that 
\begin{align}
	\dim(\SCL_{l,\al_l} \oli\SE_l(\oli f_l)) = |\oli\SE_l| + |\oli\SD_l| 
\end{align} 
so that the order $\al_l$ of possible Taylor subtractions is restricted by
\begin{align}
	\min\{\de^{(1)}_{\oli V_l}, \de^{(2)}_{\oli V_l}\} < |\oli\SE_l| + |\oli\SD_l| + |\al_l| \leq \max\{\de^{(1)}_{\oli V_l}, \de^{(2)}_{\oli V_l}\}.
\end{align}
\end{proof}
\begin{rmk}
Note that the correction terms are again expressed in local field monomials and the transition to time-ordered products without the connectedness condition is defined recursively. In particular, all quantities are well-defined distributions over the whole domain.
\end{rmk}
Using the derivation of Theorem \ref{th:GZI}, we may rewrite Proposition \ref{le:ZI} in terms of time-ordered products.
\begin{cor}\label{co:ZI}
	The change of subtraction degree for one distinguished monomial $N_{\de}[\Ph(f)]$, with $\de\in\{a,b\}$ and $a>b\geq\dim(\Ph)$, in a time-ordered product 
	\begin{align}
		\CT^\mathrm{conn}_{R} \left\{ N_{\de}[\Ph(f)] \prod_{i=1}^n \Ph_i(f_i) \right\}
	\end{align}
	is given by 
	\begin{multline}
		 \CT^\mathrm{conn}_{R} \left\{ N_{a}[\Ph(f)] \prod_{i=1}^n \Ph_i(f_i)  \right\} = \CT^\mathrm{conn}_{R} \left\{ N_{b}[\Ph(f)] \prod_{i=1}^n \Ph_i(f_i) \right\} \\
		- \sum_\SV \sum_\SE \sum_{\al} \frac{1}{\al!} \CT^\mathrm{conn}_{R} \left\{ \prod_{k\in\oli\SV} \Ph_k(f_k) N_a[D^\al \oli\SE(\oli f)] \right\} \times \\
		\times \CT^\mathrm{conn}_{R}\left\{ \sum_{\al'\cup\al''=\al} (x_0-\oli x_\SV)^{\al'} N_b[\Ph/\oli\SE(f)] \prod_{l\in\SV\setminus\{V_0\}} (x_l-\oli x_\SV)^{\al''} (\Ph_l/\oli\SE)(f_l)]\right\}
	\end{multline}
	with
	\begin{align}
		b < |\oli\SE| + |\oli\SD| + |\al| \leq a.
	\end{align}
\end{cor}
From Theorem \ref{th:GZI}, we read off that the subtraction degree is purely determined by the involved elementary fields $\ph$ and their covariant derivatives. 
This implies that the actual structure of subgraphs does not play a role for the renormalization prescription and it is sufficient to know the number and degree of external legs. 
But if we discard the structure of subgraphs, any subtraction degree of a renormalization part is affected by the chosen subtraction degree of sub-renormalization parts. 
Namely, consider an element $\g$ of a (saturated) forest $F$, where $\g_1,...,\g_c\in F$ are maximal with respect to $\g$. Then we have
\begin{align}
	\uvd(\g) = \uvd(\oli\g) + \sum_{i=1}^c \uvd(\g_i)
\end{align}
so that any $\de(\g_i)>\uvd(\g_i)$ breaks the renormalization prescription. 
But we may restore it by demanding that the inequality
\begin{align}
	\de(\g) \geq \de(\oli\g) + \sum_{i=1}^c \de(\g_i)
\end{align}
holds recursively.

\section{Normal Products}

We pointed out in the beginning of this work that changes of subtraction degrees may occur in limits of coinciding vertices of graph weights, thus require the application of the Zimmermann identity in Theorem \ref{th:GZI}. 
Furthermore we noticed that the limits may create new renormalization parts, join multiply connected components and induce overlap in existing forests. 
The latter does not appear in the momentum space treatment due to a deviating definition of renormalization parts and, therefore, in our approach requires an entirely new analysis. 
Let us discuss the various cases first. Suppose that there are several connected components. Then vertices taking part in the limit are either in one connected component or distributed over several components, not necessarily containing renormalization parts. 
By construction of the $R$-operation in configuration space, we may only exclude graphs, which consist of only those vertices which are not involved in the limit, since this corresponds to Wick ordering essentially. 
Clearly, we can define the weight $R_\De u_0^\e[\De]$ after performing the limit, i.e. $\G\rightarrow\De$. 
Assume that $\G= \hat\De$ consists of connected components $\G_1,...,\G_c$ such that
\begin{align}
	R_\G u_0^\e[\G] = \prod_{j=1}^c R_{\G_j} u_0^\e[\G_j].
\end{align}
Note that only vertices of the limit being in one connected component can induce the application of the Zimmermann identity, but overlap in existing forests can be created also in the case of multiple connected components. (see Figure \ref{fig:OverlapandZI})
\begin{figure}[h]
  \centering
	\begin{center}
	\begin{subfigure}[b]{0.45\textwidth}
    \centering
		\includegraphics[width=0.7\textwidth]{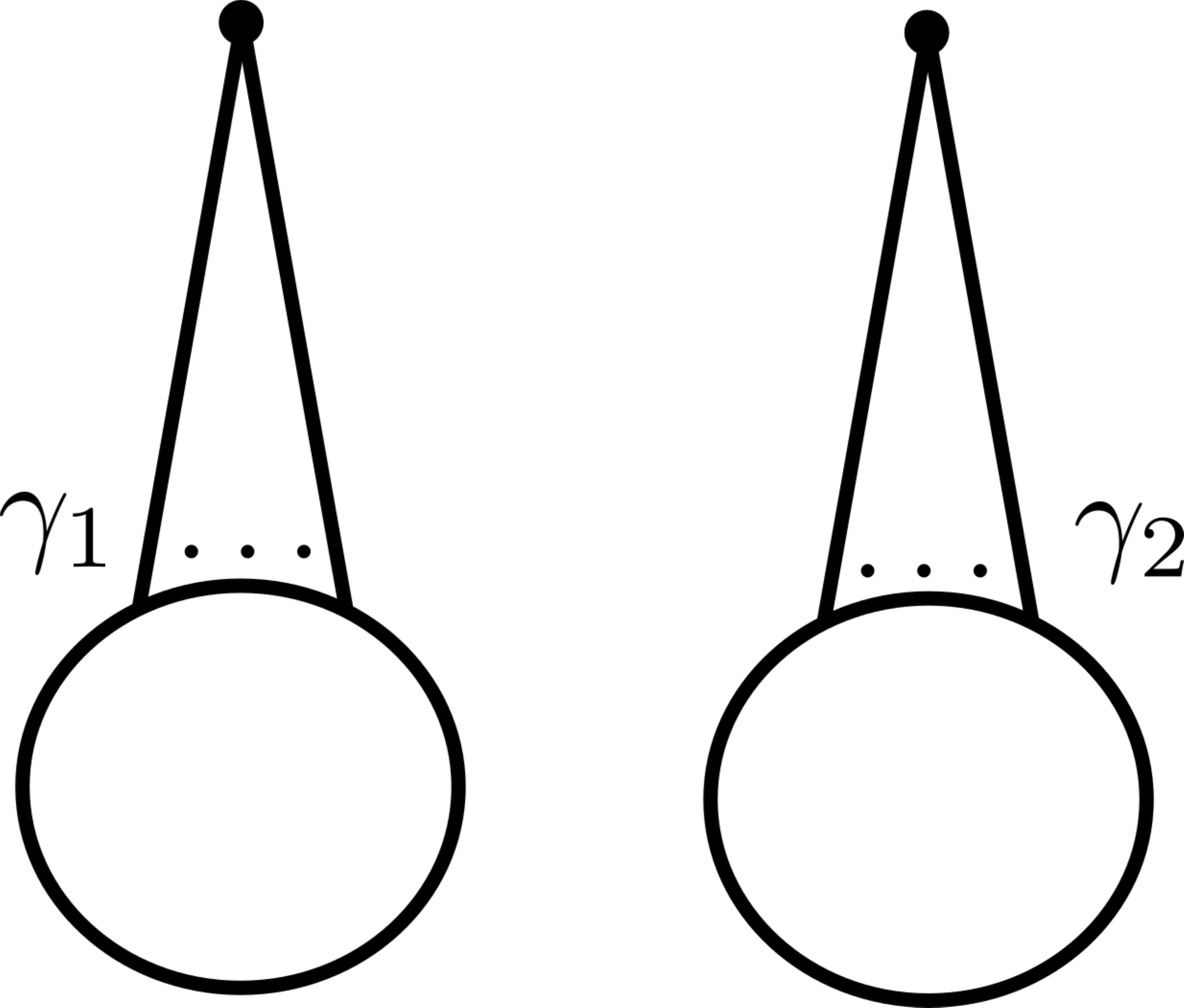}
		\caption{Disjoint renormalization parts.}
		\label{fig:2disjointRparts}	 	
	\end{subfigure}
	\begin{subfigure}[b]{0.45\textwidth}
    \centering
		\includegraphics[width=0.7\textwidth]{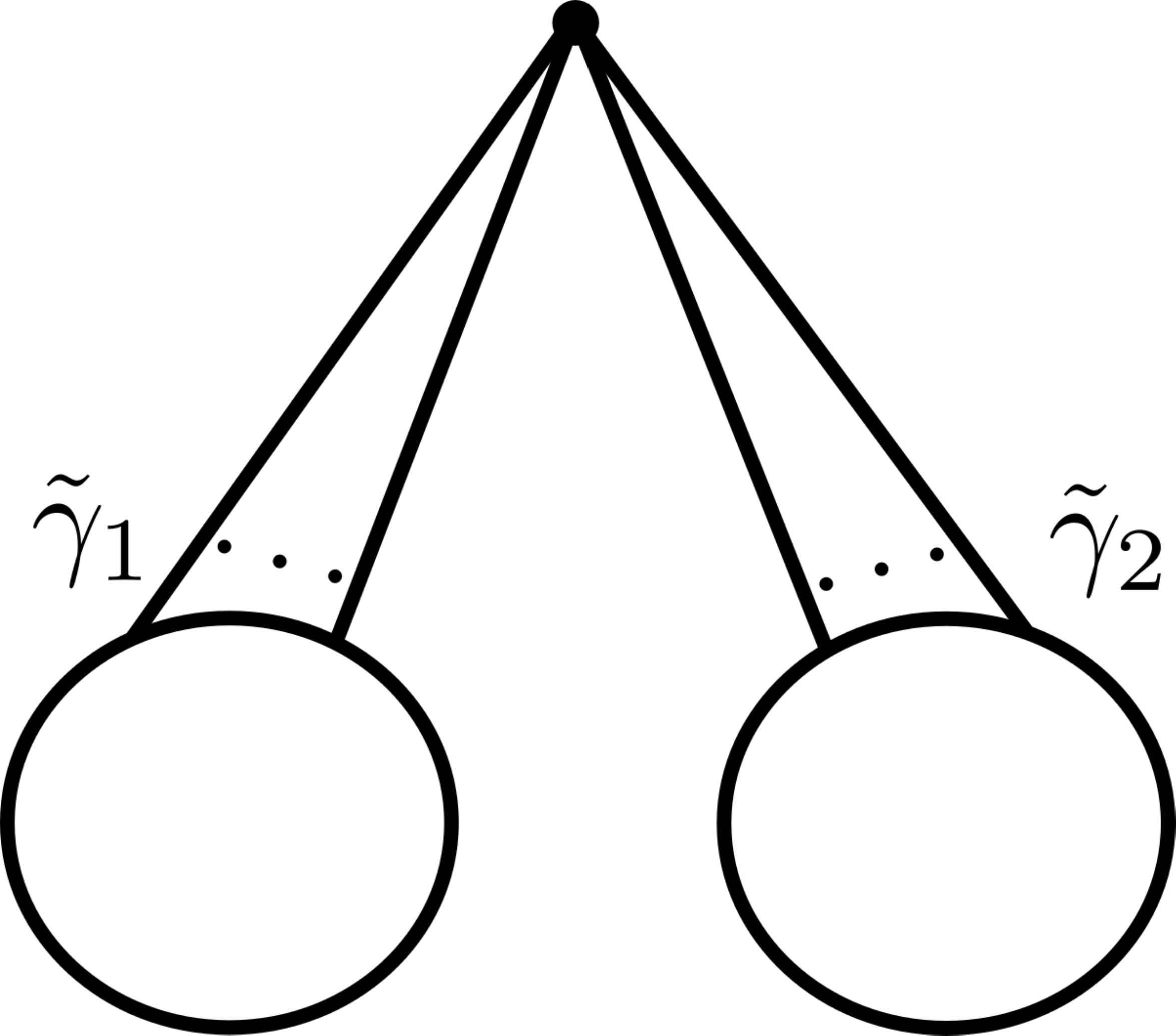}
		\caption{Overlapping renormalization parts.}
		\label{fig:ovelapRparts}	 	
	\end{subfigure}

	\begin{subfigure}[b]{0.45\textwidth}
    \centering
		\includegraphics[width=0.7\textwidth]{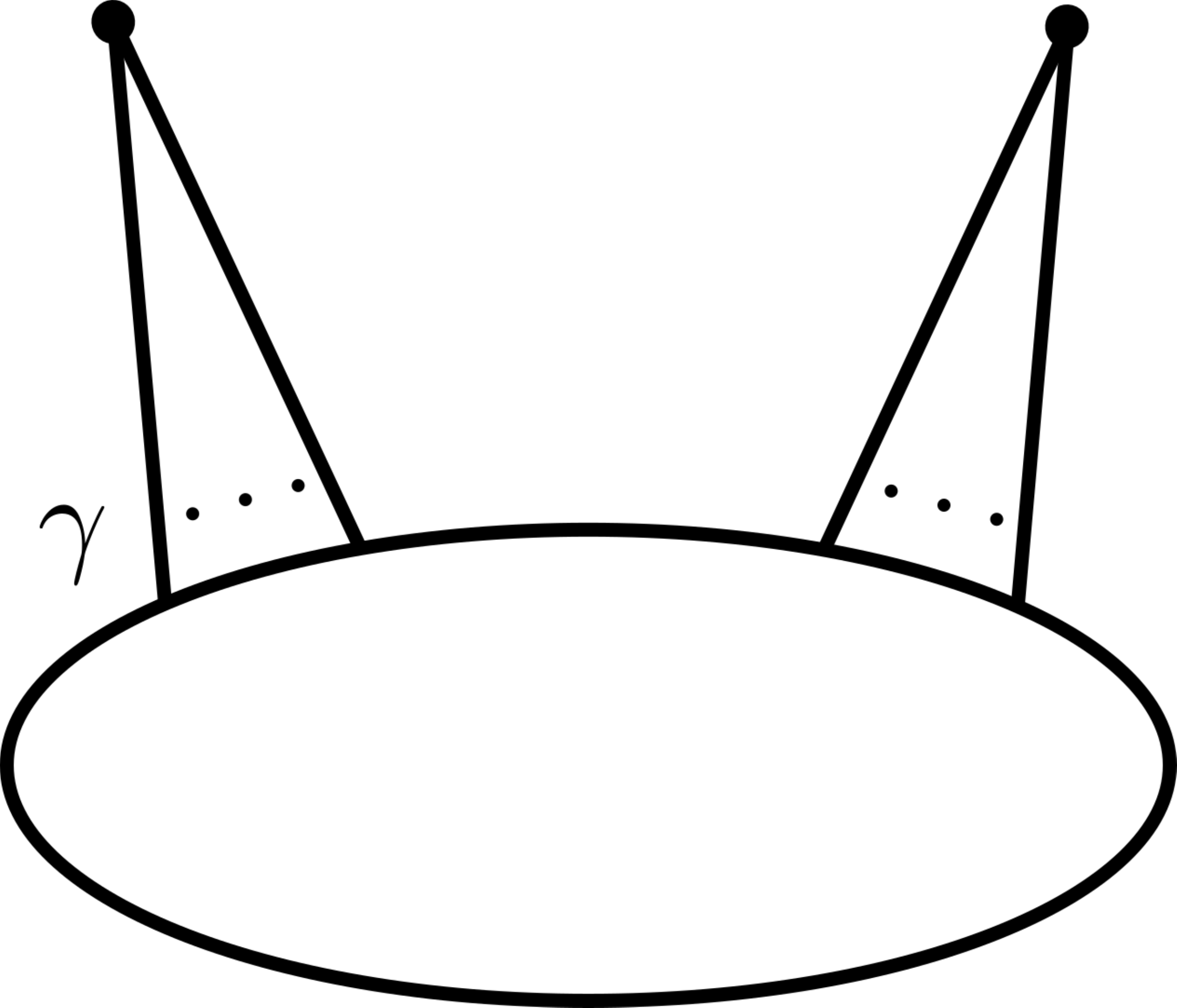}
		\caption{Renormalization part involving two limit vertices.}
		\label{fig:bigRpart}	 	
	\end{subfigure}
	\begin{subfigure}[b]{0.45\textwidth}
    \centering
		\includegraphics[width=0.7\textwidth]{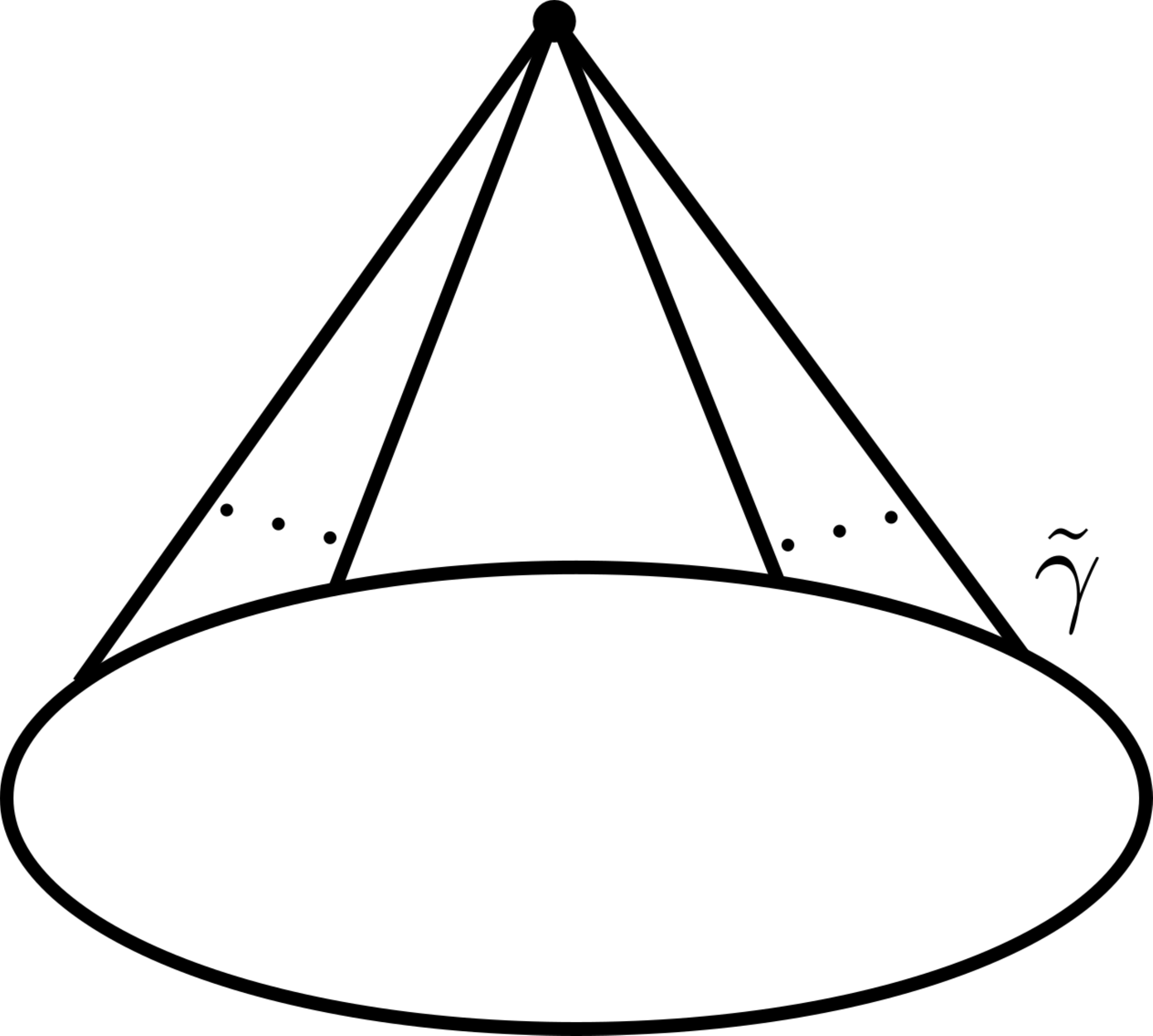}
		\caption{Renormalization part with increased subtraction degree.}
		\label{fig:ZIRpart}	 	
	\end{subfigure} 
	\end{center}
	\caption{Illustration of overlap creation before (a) and after (b) the coincidence limit, respectively, as well as an increase of subtraction degree and the necessity of the Zimmermann identity in the transition from (c) to (d). }
	\label{fig:OverlapandZI}	
\end{figure}

We begin with the case of two monomials coinciding in the limit. 
Apart from didactical purposes, the reason for this lies in the fact that one may iterate this limit in order to arrive at a limit of more than two vertices. 
But this approach leaves us with the problem whether all sequences of partial limits lead to the same result. 
This question of \emph{associativity} will not be answered in this work. 
Nevertheless consider two monomials $\Ph_1(f_1)$ and $\Ph_2(f_2)$ inserted into the time-ordered product
\begin{align}
	\CT_{R}^\mathrm{conn} \left\{ N_{\de_1}[\Ph_1(f_1)] N_{\de_2}[\Ph_2(f_2)] \prod_{j=3}^n N_{\de_j}[\Ph_j(f_j)] \right\}
\end{align}
restricted to contributions with only one connected component $\G$. 
We associate $V_1\in V(\G)$ and $V_2\in V(\G)$ to $\Ph_1(f_1)$ and $\Ph_2(f_2)$, respectively. 
If there exist non-overlapping renormalization parts $\g_1$ with $V_1\in V(\g_1)$ and $\g_2$ with $V_2\in V(\g_2)$, then $\widetilde{\g_1\cup\g_2}$ is a renormalization part in $\De$ but also $\tilde \g_1$ and $\tilde \g_2$ are overlapping. 
Further, renormalization parts $\g_{12}$ with $V_1,V_2\in V(\g_{12})$ require the application of the Zimmermann identity in the transition to $\tilde \g_{12}$. 
Since 
\begin{align}
	\lim_{V_1,V_2\rightarrow V_0} \G = \tilde \G \bydef \De,
\end{align}
the limit for the weight $u_0^\e[\G]$ is only extendible a priori for $R_\De u_0^\e[\G]$. At this point, it becomes clear why we defined the subtraction point to be
\begin{align}
	\oli x_\g = \frac{1}{2|E(\g)|} \sum_{v\in V(\g)} |E(\g|v)| x_v
\end{align}
instead of the standard mean coordinate. 
Namely, the standard mean coordinate is discontinuous in the limit of coinciding arguments and hence does not allow for a comparison of subtractions in $R_\De u_0^\e[\G]$ and $R_\G u_0^\e[\G]$. 
Since the Taylor operators act on the line complement, we use $V(\g)$ and $V(\tilde \g)$ synonymously. 
In the work of Zimmermann \cite{Zimmermann:1972tv}, they would coincide because the vertices (single field operators) would correspond to external legs in $\G$. 
In our treatment, instead, we may find external legs attached to $V_1$ and $V_2$. 
Before deriving the corrections terms, we want to state a relation which is conjugate to an observation of Zimmermann \cite{Zimmermann:1972tv}. 
Due to the definition of the $R$-operation, a renormalization part itself remains unchanged and we find for renormalization parts $\si\in\G$ and $\ta\bydef\tilde\si$
\begin{align}\label{eq:GammaToDeltaAndBack}
	u_0^\e[\G] & = u_0^\e[\G\setminus\si] u_0^\e[\si] = u_0^\e[\G\setminus\tilde\si] u_0^\e[\si] = u_0^\e[\De\setminus\ta] u_0^\e[\si] \\
	u_0^\e[\De] & = u_0^\e[\G\setminus\si] u_0^\e[\ta].
\end{align} 
\begin{prop}\label{pr:NormalProduct2Vertex}
	The difference 
	\begin{multline}
		\CT_{R_\G}^\mathrm{conn} \left\{ N_{\de_1}[\Ph_1(f_1)] N_{\de_2}[\Ph_2(f_2)] \prod_{j=3}^n N_{\de_j}[\Ph_j(f_j)] \right\} -\\- \CT_{R_\De}^\mathrm{conn} \left\{ N_{\de_1}[\Ph_1(f_1)] N_{\de_2}[\Ph_2(f_2)] \prod_{j=3}^n N_{\de_j}[\Ph_j(f_j)] \right\},
	\end{multline}
	with $\De$ referring to forests constructed as if $f_1=f_2$ and $\G$ referring to the general setting, is given by
	\begin{multline} 
		\sum_{\SV_1,\SV_2} \sum_{\SE_1,\SE_2} \sum_{\al_1} \sum_{\al_2} \frac{1}{\al_1!\al_2!} \CT^\mathrm{conn}_{R} \Bigg\{ \prod_{k\in\oli \SV} N_{\de_k}[\Ph_k(f_k)] N_{\de_{\oli V_1}}\times \\
		\times [\SCL_{\al_1}D^{\al_1} \oli\SE_1(\oli f_1; f_{\SV_1})] N_{\de_{\oli V_2}}[\SCL_{\al_2}D^{\al_2} \oli\SE_2(\oli f_2; f_{\SV_2})] \Bigg\} \\
		+ \sum_{\SV_{12}} \sum_{\SE_{12}} \sum_{\al} \frac{1}{\al!} \CT^\mathrm{conn}_{R} \left\{ \prod_{k\in\oli \SV_{12}} N_{\de_k}[\Ph_k(f_k)] N_{\de_{\oli V_{12}}}[\SCL_\al D^\al \oli\SE_{12}(\oli f_{12}; f_{\SV_{12}})] \right\},
	\end{multline}
	where 
	\begin{align}
		& 0 < |\oli\SE_{1/2}| + |\oli\SD_{1/2}| + |\al_{1/2}| \leq \de_{\oli V_{1/2}}, \\
		& \de_{\widehat{\oli V}_{12}} < |\oli\SE_{12}| + |\oli\SD_{12}| + |\al| \leq \de_{\oli V_{12}}, \\
		\intertext{and}
		& \SCL_{\al_{1/2}} = \CT^\mathrm{conn}_{R} \left\{ \prod_{l\in\SV_{1/2}} (x_{l}-\oli x_{1/2})^{\al_{1/2}} N_{\de_{l}}[\Ph_{l}/\oli\SE_{1/2}(f_{l})] \right\},\\
		& \SCL_\al = \CT^\mathrm{conn}_{R} \left\{ \prod_{l\in\SV_{12}} (x_l-\oli x_{12})^\al N_{\de_l}[\Ph_l/\oli\SE_{12}(f_l)] \right\}.
	\end{align}
\end{prop}
\begin{proof}
We begin our computation by the decomposition of the forest formula for $\De$, i.e.
\begin{align}\label{eq:pOPEGraphTwoVertex}
	R_\De u_0^\e[\G] & = \sum_{F\in\SF_\De} \prod_{\g\in F} (-t(\g)) u_0^\e[\G] \\
	& = \underbrace{\sum_{F_\perp \in \SF_\perp} \prod_{\g\in F_\perp} (-t(\g)) u_0^\e[\G]}_{\neq R_\G u_0^\e[\G] \mbox{(in general)}} + \underbrace{\sum_{F_0\in\SF_0} \prod_{\g\in F_0} (-t(\g)) u_0^\e[\G]}_{\bydef X_\De u_0^\e[\G]}.
\end{align}
The set $\SF_0$ contains all $\De$-forests such that there exists a renormalization part containing $V_0$ in its vertex set. 
In analogy to the derivation of the Zimmermann identity of Lemma \ref{le:ZI}, we can find for any $F_0\in\SF_0$ a minimal graph $\ta$ containing $V_0$ such that
\begin{align}\label{eq:XDeltaTwoVerticesInitial}
	X_\De u_0^\e[\G] = \sum_{\ta\in\ST_\De} \sum_{\oli{F}_\ta \in \oli{\SF}_\ta} \prod_{\g\in\oli{F}_\ta} (-t(\g)) (-t(\ta)) \sum_{\ul{F}_\ta \in \ul{\SF}_\ta} \prod_{\g'\in \ul{F}_\ta} (-t(\g')) u_0^\e[\G]
\end{align}
Note that $\si \bydef \hat \ta$ may or may not be a renormalization part of $\G$, i.e. an element of $\G$-forests in $\SF_\G$. 
We take this into account by replacing $(-t(\ta))$ in \eqref{eq:XDeltaTwoVerticesInitial} with $(-(t(\ta)-t(\si)))$ and the convention that $t(\si)=0$ if $\si$ is not a renormalization part for $\G$ so that 
\begin{align}
	X_\De u_0^\e[\G] & = \sum_{\ta\in\ST_\De} \sum_{\oli{F}_\ta \in \oli{\SF}_\ta} \prod_{\g\in\oli{F}_\ta} (-t(\g)) (-(t(\ta)-t(\si))) \sum_{\ul{F}_\ta \in \ul{\SF}_\ta} \prod_{\g'\in \ul{F}_\ta} (-t(\g')) u_0^\e[\G] \\
	& = - \sum_{\ta\in\ST_\De} \sum_{|\al|=\de(\si)+1}^{\de(\ta)} \frac{1}{\al!} R_{\De/\ta}(D^\al_{\oli x_{\ta}} u_0^\e[\De/\ta]) (x-\oli x_\ta)^\al R_{\ta^\perp} u_0^\e[\si]
\end{align}
We observe that due to \eqref{eq:GammaToDeltaAndBack}, we have
\begin{align}
	R_{\De/\ta}(D^\al_{\oli x_\ta} u_0^\e[\De/\ta]) = R_{\G/\si}(D^\al_{\oli x_\si} u_0^\e[\G/\si])
\end{align}
but 
\begin{align}\label{eq:RtauRsigma}
	R_{\ta^\perp}u_0^\e[\si] \neq R_{\si^\perp} u_0^\e[\si]
\end{align}
holds in general. 
If the two terms in \eqref{eq:RtauRsigma} were to be equal, then we would also obtain
\begin{align}
	\sum_{F_\perp \in \SF_\perp} \prod_{\g\in F_\perp} (-t(\g)) u_0^\e[\G] = R_\G u_0^\e[\G]
\end{align}
in \eqref{eq:pOPEGraphTwoVertex}. 
Suppose that we have
\begin{align}
	R_\G u_0^\e[\G] = \sum_{F_\perp \in \SF_\perp} \prod_{\g\in F_\perp} (-t(\g)) u_0^\e[\G] + X_\G u_0^\e[\G],
\end{align}
where $X_\G u_0^\e[\G]$ contains all $\G$-forests which lose the poset structure in the limit of $V_j\rightarrow V_0$. 
Let $\SZ$ be the set pairs $(\z_1,\z_2)$ of mutually disjoint renormalization parts of $\G$, which contain $V_1$ or $V_2$, respectively. 
Then we obtain
\begin{align}
	X_\G u_0^\e[\G] & = \sum_{(\z_1,\z_2)\in \SZ} \sum_{\oli F_\z \in \oli \SF_\z} \sum_{\ul F_\z \in \ul \SF_\z} \prod_{\g\in\oli F_\z} (-t(\g)) (-t(\z_1)) (-t(\z_2)) \prod_{\g'\in \ul F_\z} (-t(\g')) u_0^\e[\G] \\
	& = \sum_{(\z_1,\z_2)\in \SZ} R_{\G/\z}((-t(\z_1)) (-t(\z_2)) u_0^\e[\G/\z]) R_{\z^\perp_1} u_0^\e[\z_1] R_{\z^\perp_2}u_0^\e[\z_2] \\
	& = \sum_{(\z_1,\z_2)\in \SZ} \sum_{|\al_1|=0}^{\de(\z_1)} \sum_{|\al_2|=0}^{\de(\z_2)} \frac{1}{\al_1!\al_2!} R_{\G/\z} (D^{\al_1}_{\oli x_{\z_1}} D^{\al_2}_{\oli x_{\z_2}} u_0^\e[\G/\z]) \times \\
	& \qquad \times (x-\oli x_{\z_1})^{\al_1} (x- \oli x_{\z_2})^{\al_2} R_{\z^\perp_1} u_0^\e[\z_1] R_{\z^\perp_2}u_0^\e[\z_2]
\end{align} 
and conclude
\begin{align}
	R_\De u_0^\e[\G] = & R_\G u_0^\e[\G] - X_\G u_0^\e[\G] + X_\De u_0^\e[\G] \\
	= & R_\G u_0^\e[\G] - \sum_{(\z_1,\z_2)\in \SZ} \sum_{|\al_1|=0}^{\de(\z_1)} \sum_{|\al_2|=0}^{\de(\z_2)} \frac{1}{\al_1!\al_2!} R_{\G/\z} (D^{\al_1}_{\oli x_{\z_1}} D^{\al_2}_{\oli x_{\z_2}} u_0^\e[\G/\z]) \times \\
	& \qquad \qquad \qquad \times (x-\oli x_{\z_1})^{\al_1} (x- \oli x_{\z_2})^{\al_2} R_{\z^\perp_1} u_0^\e[\z_1] R_{\z^\perp_2}u_0^\e[\z_2]\\
	& - \sum_{\ta\in\ST_\De} \sum_{|\al|=\de(\si)+1}^{\de(\ta)} \frac{1}{\al!} R_{\De/\ta}(D^\al_{\oli x_\ta} u_0^\e[\De/\ta]) (x-\oli x_\ta)^\al R_{\ta^\perp} u_0^\e[\si] .
\end{align}
Analogously to Theorem \ref{th:GZI}, we sum over all graphs $\G$
\begin{align}
	\sum_\G & ( R_\G u_0^\e[\G] - R_\De u_0^\e[\G] ) \\ 
	= & \sum_{\SV_1,\SV_2} \sum_{\SE_1,\SE_2} \sum_{\al_1} \sum_{\al_2} \frac{1}{\al_1!\al_2!} \CT^\mathrm{conn}_{R} \left\{ \prod_{k\in\oli \SV} N_{\de_k}[\Ph_k(f_k)] N_{\de_{\oli V_1}}[D^{\al_1} \oli\SE_1(\oli f_1)] N_{\de_{\oli V_2}}[D^{\al_2} \oli\SE_2(\oli f_2)] \right\} \times \nonumber \\
	& \quad\times \underbrace{\CT^\mathrm{conn}_{R} \left\{ \prod_{l\in\SV_1} (x_l-\oli x_{1})^{\al_1} N_{\de_l}[\Ph_l/\oli\SE_1(f_l)] \right\}}_{\simeq \SCL_{\al_1}(\oli x_{1}; x_{\SV_1})} \underbrace{\CT^\mathrm{conn}_{R} \left\{ \prod_{l'\in\SV_2} (x_{l'}-\oli x_{2})^{\al_2} N_{\de_{l'}}[\Ph_{l'}/\oli\SE_2(f_{l'})] \right\}}_{\simeq\SCL_{\al_2}(\oli x_{2}; x_{\SV_2})} \\
	& + \sum_{\SV_{12}} \sum_{\SE_{12}} \sum_{\al} \frac{1}{\al!} \CT^\mathrm{conn}_{R} \left\{ \prod_{k\in\oli \SV_{12}} N_{\de_k}[\Ph_k(f_k)] N_{\de_{\oli V_{12}}}[D^\al \oli\SE_{12}(\oli f_{12})] \right\} \times \nonumber \\
	& \qquad \qquad \qquad \qquad \qquad \qquad \qquad \qquad \times \underbrace{\CT^\mathrm{conn}_{R} \left\{ \prod_{l\in\SV_{12}} (x_l-\oli x_{12})^\al N_{\de_l}[\Ph_l/\oli\SE_{12}(f_l)] \right\}}_{\simeq\SCL_\al(\oli x_{12}; x_{\SV_{12}})},
\end{align}
where $x_{\SV_{\bullet}}$ indicates the dependence on all arguments assigned to vertices in $\SV_{\bullet}$ and the derivatives from Taylor subtractions are restricted by
\begin{align}
	0 < |\oli\SE_{1/2}| + |\oli\SD_{1/2}| + |\al_{1/2}| \leq \de_{\oli V_{1/2}}
\end{align}
for the corrections from $\sum_\G X_\G u_0^\e[\G]$ and
\begin{align}
	\de_{\widehat{\oli V}_{12}} < |\oli\SE_{12}| + |\oli\SD_{12}| + |\al| \leq \de_{\oli V_{12}}
\end{align}
for the corrections from $\sum_\G X_\De u_0^\e[\G]$, where $\de_{\widehat{\oli V}_{12}} > 0$ holds only in the case of applied Zimmermann identity. 
This proves the assertion.
\end{proof}
We remark that the corrections $\sum_\G X_\G u_0^\e[\G]$ do not appear in the derivation of \cite{Zimmermann:1972tv} and arise from the differing definition of renormalization parts. 
In particular, those corrections do not have the desired form, i.e. they are expressed by two instead of just one local field monomial so that, in the limit of coinciding arguments, singularities occur on the level of distributions instead of functions $\SCL$. 
Nevertheless we may apply Proposition \ref{pr:NormalProduct2Vertex} to $\sum_\G X_\G u_0^\e[\G]$ and obtain corrections in the desired form after a finite number of iterations since we considered only time-ordered products with finitely many field monomials.

Next, we generalize Proposition \ref{pr:NormalProduct2Vertex} to the case of $n$ monomials in the coincidence limit. 
Let us first look at a time-ordered product
\begin{align}
	\CT_{R} \left\{ \prod_{j=1}^n N_{\de_j}[\Ph_j(f_j)] \right\},
\end{align}
where all monomials take part in the limit of coinciding arguments. 
Expanding as usual into graphs, the resulting graph $\De$ after the limit is a bouquet graph, i.e. a graph with one vertex and edges $e$ with $s(e)=t(e)$, which is both vanishing under action of the $R$-operation and regularized by Wick-ordering of the whole product of monomials. 
But parts of bouquet graphs emerge also in the treatment of the time-ordered products with two sets of monomials
\begin{align}\label{eq:TimeOrdNVertexCase}
	\CT_{R} \Bigg\{ \underbrace{\prod_{i=1}^n N_{\de_i}[\Ph_i(f_i)]}_{\mathrm{limit-vertices}} \prod_{j=1}^m N_{\de_{n+j}}[\Ph_{n+j}(f_{n+j})] \Bigg\},
\end{align}
where only the first $n$ monomials are affected by the limit. 
Again, we emphasize that contributions from tadpoles vanish after applying the $R$-operation. 
Due to this, we restrict \eqref{eq:TimeOrdNVertexCase} to contributions without edges among limit-vertices after the application of Wick's theorem and indicate this by
\begin{align}
	\prod_{i=1}^n N_{\de_i}[\Ph_i(f_i)] \stackrel{\mathrm{Wick}}{\longrightarrow} : \prod_{i=1}^n N_{\de_i}[\Ph_i(f_i)] :.
\end{align}
Expanding \eqref{eq:TimeOrdNVertexCase} into graphs $\G$ and setting as above $\De \bydef \tilde \G$, we have to compare $R_\De u_0^\e[\G]$ with $R_\G u_0^\e[\G]$ again. 
In this case the comparison is significantly more involved since we may find several limit-vertices at one connected component of $\G$ and we may have several connected components of $\G$ each containing at least one limit-vertex. 
Hence suppose that $\G$ has $\G_1,..., \G_k$ connected components, where each $\G_j$ contains $n_j$ limit-vertices with $\sum_j n_j = n$. 
Analogously to the 2-vertex case, we expect to obtain corrections from $\De$-forests as well as $\G$-forests, i.e.
\begin{align}
	R_\De u_0^\e[\G] = R_\G u_0^\e[\G] - X_\G u_0^\e[\G] + X_\De u_0^\e[\G].
\end{align}
For the corrections regarding overlap creation in forests, induced by the limit, we observe that those terms consist of all sets $\{\z\}_c$ of $c$ mutually disjoint renormalization parts $\z$ in $\G$, where each renormalization part contains at least one limit-vertex. 
We find sets with $2\leq c \leq n$ and subsume those in $\SZ_c$ so that
\begin{align}
	X_\G u_0^\e[\G] & = \sum_{c=2}^n \sum_{\{\z\}_c \in \SZ_c} \sum_{\oli F_{\{\z\}_c} \in \oli \SF_{\{\z\}_c}} \sum_{\ul F_{\{\z\}_c} \in \ul \SF_{\{\z\}_c}} \prod_{\g\in\oli F_{\{\z\}_c}} (-t(\g)) \prod_{l=1}^k (-t(\z_l)) \prod_{\g'\in \ul F_{\{\z\}_c}} (-t(\g')) u_0^\e[\G] \\
	& = \sum_{c=2}^n (-1)^c \sum_{\{\z\}_c \in \SZ_c} \sum_{\al_1,...,\al_c} \frac{1}{\al_1!...\al_c!} R_{\G/\{\z\}_c} (D^{\al_1}_{\oli x_{\z_1}}...D^{\al_c}_{\oli x_{\z_c}} u_0^\e[\G/\{\z\}_c]) \times \nonumber \\
	& \qquad \qquad \qquad\qquad \qquad \qquad \qquad \qquad \qquad\qquad \qquad \times \prod_{l=1}^c (x_l - \oli x_{{\z_l}})^{\al_l} R_{\z_l^\perp}u_0^\e[\z_l],
\end{align} 
where the sums in $\al_j$ run from $0$ to $\de(\z_j)$. 
The contributions from all graphs are then given by
\begin{multline}
	\sum_\G X_\G u_0^\e[\G] = \sum_{c=2}^n \sum_{\{\SV\}_c} \sum_{\{\SE\}_c} \sum_{\al_1,...,\al_c} \frac{1}{\al_1!...\al_c!} \CT_{R}\left\{ \prod_{k\in\oli\SV} N_{\de_k}[\Ph_k(f_k)] \prod_{l=1}^c N_{\de_{\oli V_l}}[D^{\al_l}\oli\SE_l(\oli f_l)] \right\} \times \\
	\times \prod_{l'=1}^c \underbrace{\CT^\mathrm{conn}_{R} \left\{ \prod_{r\in \SV_{l'}} (x_r-\oli x_{l'})^{\al_{l'}} N_{\de_r} [(\Ph_r/\oli\SE_{l'})(f_r)] \right\}}_{\SCL_{\al_{l'}}(\oli x_{l'}; x_{\SV_{l'}})},
\end{multline}
where
\begin{align}
	0 < |\oli\SE_j| + |\oli\SD_j| + |\al_j| \leq \de_{\oli V_l}.
\end{align}
It is left to examine the correction terms related to the introduction of new renormalization parts or the increase of subtraction degree induced by the limit, respectively. Note that those terms can only appear once per forest $F$ if we require them to be minimal in $F$. Therefore we obtain
\begin{align}
	X_\De u_0^\e[\G] & = \sum_{\ta\in\ST_\De} \sum_{\oli F_\ta \in \oli \SF_\ta} \sum_{\ul F_\ta \in \ul \SF_\ta} \prod_{\g\in\oli F_\ta} (-t(\g)) (-(t(\ta) - t(\si))) \prod_{\g'\in \ul F_\ta} (-t(\g')) u_0^\e[\G] \\
	& = - \sum_{\ta\in\ST_\De} \sum_{|\al|=\de(\si)+1}^{\de(\ta)} \frac{1}{\al!} R_{\De/\ta}(D^\al_{\oli x_\si} u_0^\e[\G\setminus\si]) (x-\oli x_{\ta})^\al R_{\ta^\perp} u_0^\e[\ta]    ,
\end{align}
where $t(\si)=t(\hat\ta)=0$ if $\si$ is not a renormalization part. 
Again, summing over all graphs $\G$ we compute
\begin{multline}
	\sum_\G X_\De u_0^\e[\G] = \sum_\SV \sum_\SE \sum_\al \frac{1}{\al!} \CT_{R} \left\{ \prod_{k\in\oli\SV} N_{\de_k}[\Ph_k(f_k)] N_{\de_{\oli V}}[D^\al \oli\SE(\oli f)] \right\} \times \\
	\times \underbrace{\CT^\mathrm{conn}_{R} \left\{ \prod_{l\in\SV} (x_l - \oli x)^\al N_{\de_l}[(\Ph_l/\oli\SE)(f_l)] \right\}}_{\SCL_\al(\oli x; x_\SV)},
\end{multline}
with
\begin{align}
	\de_{\widehat{\oli V}} < |\oli\SE| + |\oli\SD| + |\al| \leq \de_{\oli V}
\end{align}
and $\de_{\widehat{\oli V}}>0$ only in the case of Zimmermann identity. 
Finally, we want to combine all contributions in the definition of normal products and perform the coincidence limit.
\begin{defi}\label{de:NormalProduct}
	A normal product of degree $\de$
	\begin{align}
		N_\de\Big[\prod_{i=1}^n \Ph_i(f_i)\Big]
	\end{align}
	with $\de = \sum_i \de_i \geq \sum_i \dim(\Ph_i)$ inserted into a time-ordered product with $m$ (spectator) monomials is defined by
	\begin{multline}
		\CT_{R} \left\{ N_\de[\prod_{i=1}^n \Ph_i(f_i)] \prod_{j=1}^m N_{\de_{n+j}}[\Ph_{n+j}(f_{n+j})] \right\} = \CT_{R} \left\{ :\prod_{i=1}^n N_{\de_j}[\Ph_i(f_i)]: \prod_{j=1}^m N_{\de_{n+j}}[\Ph_{n+j}(f_{n+j})] \right\} \\
		+ \sum_\SV \sum_\SE \sum_\al \frac{1}{\al!} \CT_{R} \left\{ \prod_{k\in\oli\SV} N_{\de_k}[\Ph_k(f_k)] N_{\de_{\oli V}}[\SCL_\al D^\al \oli\SE(\oli f; f_\SV)] \right\} \\
		- \sum_{c=2}^n \sum_{\{\SV\}_c} \sum_{\{\SE\}_c} \sum_{\al_1,...,\al_c} \frac{1}{\al_1!...\al_c!} \CT_{R}\left\{ \prod_{k\in\oli\SV} N_{\de_k}[\Ph_k(f_k)] N_{\oli\de}\Big[\prod_{l=1}^c \SCL_{\al_{l}} D^{\al_l}\oli\SE_l(\oli f_l; f_{\SV_{l}})\Big] \right\},
	\end{multline}
	where
	\begin{align}
		\oli\de & \bydef \sum_{l=1}^c \de_{\oli V_l} \\
		\de_{\widehat{\oli V}} & < |\oli\SE| + |\oli\SD| + |\al| \leq \de_{\oli V} \\
		0 & < |\oli\SE_j| + |\oli\SD_j| + |\al_j| \leq \de_{\oli V_l} .
	\end{align}
\end{defi}
\begin{thm}
	Let $\Ph_1,...,\Ph_n$ be field monomials with scaling dimensions $\de_1,...,\de_n$. Then
	\begin{align}
		\lim_{f_j\rightarrow f} N_\de[\Ph_1(f_1)...\Ph_n(f_n)] & = N_\de[(\Ph_1...\Ph_n)(f)] \\
		\de \geq \sum_{i=1}^n \de_i
	\end{align}
	if inserted to time-ordered products.
\end{thm}
\begin{proof}
	Given any time-ordered product. Inserting the normal product $N_\de[\Ph_1(f_1)...\Ph_n(f_n)]$ and expanding in graphs, we obtain
	\begin{align}
		\lim_{x_j\rightarrow x} (R_\G u_0^\e[\G] - X_\G u_0^\e[\G] + X_\De u_0^\e[\G]) = \lim_{x_j
		\rightarrow x} R_\De u_0^\e[\G],
	\end{align}
	where $\De = \tilde \G$. The insertion of $N_\de[(\Ph_1...\Ph_n)(f)]$ gives
	\begin{align}
		R_\De u_0^\e[\De]
	\end{align}
	and thus we arrive at
	\begin{align}
		\lim_{x_j\rightarrow x} R_\De u_0^\e[\G] - R_\De u_0^\e[\De] = 0.
	\end{align}
\end{proof}
We could continue and calculate products of normal products inserted into time-ordered products or relate normal products of different degree following the ideas in Theorem \ref{th:GZI}. 
Both calculations are feasible but tedious and not performed in this work. 

\section{Field Equation}

In the construction of Wick monomials and time-ordered products, we transfered to off-shell fields $\ph$, i.e. we did not demand that monomials of the type $N_\de[\Ph P\ph]$ vanish identically, where $P$ denotes the wave operator. 
The techniques developed in the Section above allow us to examine such monomials in detail. 
Therefore let us study
\begin{align}\label{eq:FieldEquationTOProduct}
	\CT_{R} \left\{ N_\de[\Ph P \ph(f)] \prod_{j=1}^n N_{\de_j}[\Ph_j(f_j)] \right\}. 
\end{align} 
Without loss of generality, we consider only a simply connected component $\G$ after the application of Wick's theorem. 
Denoting the vertex of the monomial $N_\de[\Ph P \ph(f)]$ by $V_0$, there exists exactly one edge $e_0\in E(\G)$, connecting $P\ph$ to a vertex $V_j\in V(\G)\setminus\{V_0\}$. But we know that
\begin{align}
	P_{V_0} H_F(x_{s(e_0)},x_{t(e_0)}) = \de(x_{s(e_0)},x_{t(e_0)})
\end{align}
and thus the vertices $V_0$ and $V_j$ in $\G$ get fused after evaluating the Dirac-$\de$-distribution resulting in a graph $\De$. Before we turn to the analysis of the change in the singularity structure in the transition from $\G$ to $\De$, we have to discuss the subtraction degree $\de$ of the monomial $\Ph P\ph$. 
In local form we have
\begin{align}
	P\ph(x) = g^{\m\n} \na_\m\na_n \ph(x) + b(x) \ph(x)
\end{align}
so that 
\begin{align}
	N_\de[\Ph P\ph(f)] = N_\de[\Ph g^{\m\n} \na_\m \na_\n \ph(f)] + N_\de[\Ph b \ph(f)].
\end{align}
However, we find $\dim(g^{\m\n} \na_\m \na_\n \ph) = \dim(b\ph) + 2$ and therefore have to relate $N_{\dim(\Ph)+3}[\Ph b\ph]$ to $N_{\dim(\Ph)+1}[\Ph b\ph]$. 
This is performed applying Corollary \ref{co:ZI} to \eqref{eq:FieldEquationTOProduct} with $a=3$ and $b=1$, i.e.
\begin{multline}
	\CT^\mathrm{conn}_{R} \left\{ N_{\dim(\Ph)+1}[\Ph b\ph(f)] \prod_{i=1}^n N_{\de_i}[\Ph_i(f_i)] \right\} = \CT^\mathrm{conn}_{R} \left\{ N_{\dim(\Ph)+3}[\Ph b\ph(f)] \prod_{i=1}^n N_{\de_i}[\Ph_i(f_i)] \right\} \\
	+ \sum_\SV \sum_\SE \sum_\al \frac{1}{\al!} \CT^\mathrm{conn}_{R} \left\{ \prod_{k\in\oli\SV} N_{\de_k}[\Ph_k(f_k)] N_{\dim(\Ph)+3} [\SCL_\al[b] D^\al \oli\SE(\oli f)] \right\}
\end{multline}
with $\dim(\Ph) + 1 < |\oli\SE| + |\al| \leq \dim(\Ph) + 3$. This allows us to work with $N_{\dim(\Ph + 3)}[\Ph P\ph(f)]$ and thus with $\de \geq \dim(\Ph)+3$ involving additional corrections from Zimmermann identities.

Next let us analyze the change of the singularity structure in the fusion process $\G \rightarrow \tilde \G \bydef \De$. 
For disjoint renormalization parts $\g_0$ with $V_0\in V(\g_0)$ and $\g_j$ with $V_j\in V(\g_j)$, surely $\widetilde{\g_0\cup\g_j}$ is a renormalization part but we also obtain overlap for $\tilde \g_0$ and $\tilde \g_j$. 
Further, there may exist renormalization parts $\g_{0j}$ already in $\G$, which change their subtraction by contracting $e_0$ but do not change the edge set on which the Taylor polynomial is computed.
To sum up, the occuring corrections resemble the result of Proposition \ref{pr:NormalProduct2Vertex} since only pairs of vertices $(V_0,V_j)$ are involved.
\begin{thm}\label{th:FieldEquation}
	The action of a wave operator $P$ appearing in a monomial $\Ph P\ph$ inserted into a time-ordered product is given by
	\begin{align}\label{eq:FieldEquationPsi}
		\CT_{R} & \left\{ N_\de[\Ph P \ph(f)] \prod_{j=1}^n N_{\de_j}[\Ph_j(f_j)] \right\} = \\
		& \sum_{j=1}^n \CT^\mathrm{conn}_{R} \left\{ N_{\de+\de_j-4}[\Ph \frac{\de}{\de\ph} \Ph_j(f_j)] \prod_{i=1,i\neq j}^n N_{\de_i}[\Ph_i(f_i)] \right\} \\
		& + \sum_{j=1}^n \sum_{\SV_0,\SV_j} \sum_{\SE_0,\SE_j} \sum_{\al_0,\al_j} \frac{1}{\al_0! \al_j!} \times \\
		& \times \CT^\mathrm{conn}_{R} \left\{ \prod_{k\in\oli\SV_{0,j}} N_{\de_k}[\Ph_k(f_k)] N_{\de_{\oli V_0}}[\SCL_{\al_0} D^{\al_0} \oli\SE_0(\oli f_0, f_{\SV_0})] N_{\de_{\oli V_j}}[\SCL_{\al_j} D^{\al_j} \oli\SE_j(\oli f_j; f_{\SV_j})] \right\} \\
		& + \sum_{j=1}^n \sum_{\SV_{0j}} \sum_{\SE_{0j}} \sum_\al \frac{1}{\al!} \CT^\mathrm{conn}_{R} \left\{ \prod_{k\in\oli\SV_{0j}} N_{\de_k}[\Ph_k(f_k)] N_{\de_{\oli V_{0j}}}[\SCL_\al D^\al \oli\SE_{0j}(\oli f_{0j};f_{\SV_{0j}})]\right\}
	\end{align}
	with conventions for multiindices $\al_\bullet$ and functions $\SCL_{\al_\bullet}$ as in Proposition \ref{pr:NormalProduct2Vertex}.
\end{thm}
\begin{proof}
Analogously to the two-vertex case for normal products, we obtain for a single graph $\G$
\begin{align}
	R_\De u_0^\e[\De] & = R_\G u_0^\e[\G] - \sum_{(\z_0,\z_j)\in \SZ} \sum_{\oli F_\z \in \oli \SF_\z} \sum_{\ul F_\z \in \ul \SF_\z} \prod_{\g\in\oli F_\z} (-t(\g)) (-t(\z_0)) (-t(\z_j)) \prod_{\g'\in \ul F_\z} (-t(\g')) u_0^\e[\G] \\ & + \sum_{\ta_{0j}\in\ST_\De} \sum_{\oli{F}_\ta \in \oli{\SF}_\ta} \prod_{\g\in\oli{F}_\ta} (-t(\g)) (-(t(\ta_{0j})-t(\hat\ta_{0j}))) \sum_{\ul{F}_\ta \in \ul{\SF}_\ta} \prod_{\g'\in \ul{F}_\ta} (-t(\g')) u_0^\e[\G].
\end{align}
Omitting the intermediate step of spelling out Taylor operators, we compute directly the sum over all contributions $\G$ using Proposition \ref{pr:NormalProduct2Vertex} and obtain
\begin{multline}
	\sum_\G X_\G u_0^\e[\G] = \sum_{j=1}^n \sum_{\SV_0,\SV_j} \sum_{\SE_0,\SE_j} \sum_{\al_0,\al_j} \frac{1}{\al_0! \al_j!} \times \\
	\times \CT^\mathrm{conn}_{R} \left\{ \prod_{k\in\oli\SV_{0,j}} N_{\de_k}[\Ph_k(f_k)] N_{\de_{\oli V_0}}[\SCL_{\al_0} D^{\al_0} \oli\SE_0(\oli f_0; f_{\SV_0})] N_{\de_{\oli V_j}}[\SCL_{\al_j} D^{\al_j} \oli\SE_j(\oli f_j; f_{\SV_j})] \right\}
\end{multline}
with
\begin{align}
	\SCL_{\al_0}(\oli f_0; f_{\SV_0}) & = \CT^\mathrm{conn}_{R} \left\{ \prod_{l\in\SV_0\setminus\{V_0\}} (x_l-\oli x_{0})^{\al_0} N_{\de_l}[(\Ph_l/\oli\SE_0)(f_l)]  N_\de[(\Ph P\ph/\oli\SE_0)(f_0)] \right\}, \\
	\SCL_{\al_j} (\oli f_j; f_{\SV_j}) & = \CT^\mathrm{conn}_{R}\left\{ \prod_{l'\in\SV_j} (x_{l'}-\oli x_{j})^{\al_j} N_{\de_{l'}}[(\Ph_{l'}/\oli\SE_j)(f_{l'})]  \right\}
\end{align}
and
\begin{align}
	0 < |\oli\SE_{0/j}| + |\oli\SD_{0/j}| + |\al_{0/j}| \leq \de_{\oli V_0 / \oli V_j}. 
\end{align}
In the same manner, we compute for the Zimmermann identity correction terms that
\begin{align}
	\sum_\G X_\De u_0^\e[\G] = \sum_{j=1}^n \sum_{\SV_{0j}} \sum_{\SE_{0j}} \sum_\al \frac{1}{\al!} \CT^\mathrm{conn}_{R} \left\{ \prod_{k\in\oli\SV_{0j}} N_{\de_k}[\Ph_k(f_k)] N_{\de_{\oli V_{0j}}}[\SCL_\al D^\al \oli\SE_{0j}(\oli f_{0j}; f_{\SV_{0j}})]\right\},
\end{align}
where 
\begin{align}
	\SCL_\al (\oli f_{0j}; f_{\SV_{0j}}) = \CT^\mathrm{conn}_{R} \left\{ \prod_{l\in\SV_{0j}\setminus\{V_0\}} (x_l-\oli x_{0j})^\al N_{\de_l}[(\Ph_l/\oli\SE_{0j})(f_l)] N_\de[(\Ph P\ph/\oli\SE_{0j})(f_0)] \right\}
\end{align}
and 
\begin{align}
	\de_{\widehat{\oli{V}}_{0j}} < |\oli\SE_{0j}| + |\oli\SD_{0j}| + |\al| \leq \de_{\oli V_{0j}}.
\end{align}
The sum over all fused graphs $\De$ gives
\begin{align}
	\sum_{j=1}^n \CT^\mathrm{conn}_{R} \left\{ N_{\de+\de_j-4}[\Ph \frac{\de}{\de\ph} \Ph_j(f_j)] \prod_{i=1,i\neq j}^n N_{\de_i}[\Ph_i(f_i)] \right\}
\end{align}
and merging all contributions we arrive at the assertion.
\end{proof}
We finish this Section with the discussion of a special case of Theorem \ref{th:FieldEquation}. Suppose that $\Ph=1$, i.e. we consider an insertion $N_\de[P\ph(f)]$. Inserted into a time-ordered product and expanded in graphs, the vertex $N_\de[P\ph(f)]$ corresponds to an external line of those graphs. It follows that neither Zimmermann identity correction terms nor overlap creation correction terms can appear so that \eqref{eq:FieldEquationPsi} reduces to
\begin{multline}
	\CT_{R} \left\{ N_\de[P \ph(f)] \prod_{j=1}^n N_{\de_j}[\Ph_j(f_j)] \right\} =
	\sum_{j=1}^n \CT^\mathrm{conn}_{R} \left\{ N_{\de+\de_j-4}[\frac{\de}{\de\ph} \Ph_j(f_j)] \prod_{i=1,i\neq j}^n N_{\de_i}[\Ph_i(f_i)] \right\}.
\end{multline}

\section{Conclusion}

In the present work, we derived the notions of normal products and Zimmermann identities in the framework of configuration space BPHZ renormalization, which, in the original formulation in momentum space, turned out to be particularly well suited for the study of structural properties of a specific theory. 
Recall that the insertion of a normal product into time-ordered functions maintains local integrability in the coincidence limit after the application of the $R$-operation due to suitably chosen subtraction degrees. 
Provided this normal product depends on a parameter of the theory like the mass, we can study the behavior of the theory under changes of the parameter by inserting the derivative of associated normal product into every time-ordered product given as a formal power series. 
Again, we emphasize that these insertions do not require any additional renormalization techniques. 
This idea of relating insertions and derivatives with respect to parameters is formalized by the action principle \cite{Lowenstein:1971jk} such that parametric differential equations \cite{Zimmermann:1979fd,Hollands:2002ux,Brunetti:2009qc} should be derivable more conveniently. \\
Another application of the action principle can be constructed in analogy to the derivation of the field equation, where the application of the wave operator fuses vertices. 
In the same manner, normal products may be manipulated by other normal products using functional derivatives, which remove the full or just a part of monomial and replace it by or add another monomial, respectively. 
Specifically, an elementary field operator in an observable may be exchanged by a transformed elementary field operator while keeping a suitably large subtraction degree. 
With this, Ward identities should be representable as insertions into the full theory such that the effect of symmetry transformations \cite{Kraus:1991cq,Kraus:1992ru} becomes easier tractable. 
We remark that symmetries of $(M,g)$, thus symmetries of the Hadamard parametrix $H$, should be restored after the limit $\e\rightarrow 0$.

Having a full theory at hand also admits a simplification of the rather bulky results on Zimmermann identity and normal products. 
Since we considered only a single time-ordered product, it was cumbersome to keep track of fields constructing renormalization parts and newly formed vertices, thus Wick monomials, in the sum over all contributing graphs. 
In a full theory, having all possible graphs at our disposal, one may blow up a single vertex to an arbitrary renormalization part, which has to be compatible with the theory and has to have external lines matching the incident lines of the initial vertex. 
This blow-up-graph is again a graph of the theory and, vice versa, one may contract any renormalization part of a given graph to a single vertex obtaining another graph of the theory. 
The latter manipulation corresponds to the application of the $R$-operation, where we observe that the $R$-operation is performed independently of the structure of the renormalization part and, further, independently of the structure complement of the renormalization part. 
Hence the sum over all graphs containing a renormalization part, which is associated to a specific fixed vertex after contraction, can be split, after application of the $R$-operation, into a sum of renormalization parts and a sum of the complement. 
This is exactly the content of Zimmermann identities but with the difference that the sums become independent after the splitting, thus can be written in independent formal power series. 
It is evident that an analogous argument holds for normal products.
For an explicit example, we refer again to \cite{Pottel:2017fwz}.\\
The Zimmermann identity should certainly be used in the aforementioned studies of concrete theories, relating insertion of differing engineering dimension, and the definition of normal products should be of benefit in the study of operator product expansion \cite{Abdesselam:2016npc}, possibly on analytic spacetimes \cite{Hollands:2006ag}. 
Recall that we specifically emphasized the limit of two coinciding vertices, which is predestined for the question of associativity \cite{Holland:2015tia} and, furthermore, one may investigate in the convergence of the operator product expansion \cite{Hollands:2011gf}. 
As a final remark we shall point out that, differently from Zimmermann's result, the normal products have to be defined recursively due to the larger class of renormalization parts. 
It remains open whether and how this recursive definition simplifies if concrete models are considered.

\subsection*{Acknowledgments}

The author would like to thank Klaus Sibold for numerous discussions. The financial support by the Max Planck Institute for Mathematics in the Sciences and its International Max Planck Research
School (IMPRS) ``Mathematics in the Sciences'' is gratefully acknowledged. 

\bibliographystyle{halpha}
\bibliography{bphzl}

\end{document}